\documentclass[prd,12pt,tightenlines,nofootinbib]{revtex4}

\usepackage[english]{babel}
\usepackage{graphicx}
\usepackage{psfrag}
\usepackage{amsmath}
\usepackage{amssymb}
\usepackage{bbm}
\usepackage{verbatim}

\newcommand{\be}{\begin{equation}}
\newcommand{\ee}{\end{equation}}
\newcommand{\bea}{\begin{eqnarray}}
\newcommand{\eea}{\end{eqnarray}}
\newcommand{\bean}{\begin{eqnarray*}}
\newcommand{\eean}{\end{eqnarray*}}

\newcommand{\irm}{{\rm i}}
\newcommand{\e}{{\rm e}}

\renewcommand{\d}{{\rm d}}
\newcommand{\cl}[1]{{\mathcal #1}}

\newcommand{\pa}{\partial}

\newcommand{\vb}[1]{\mathbf{#1}}

\newcommand{\ds}{\displaystyle}

\newcommand{\pdiff}[2]{\frac{\partial #1}{\partial #2}}

\newcommand{\bN}{\mathbb{N}}

\newcommand{\bR}{\mathbb{R}}

\newcommand{\clA}{\cl{A}}

\newcommand{\eq}[1]{(\ref{#1})}
\renewcommand{\sec}[1]{sec.\ \ref{#1}}

\newcommand{\fig}[1]{Fig.\ \ref{#1}}

\newcommand{\tr}{{\rm tr}}
\renewcommand{\Re}{{\rm Re}}

\newcommand{\sgn}{\mathrm{sgn}}

\newcommand{\pic}[4]
{
 \begin{figure}
 \begin{center}
 \includegraphics[height=#3]{#4}
 \end{center}
 \caption{\label{#1} #2}
 \end{figure}
}

\newtheorem{theorem}{Theorem}[section]

\newtheorem{proposition}[theorem]{Proposition}

\newtheorem{definition}[theorem]{Definition}

\newenvironment{proof}[1][Proof]{\begin{trivlist}
\item[\hskip \labelsep {\bfseries #1}]}{\end{trivlist}}

\newcommand{\qed}{\nobreak \ifvmode \relax \else
      \ifdim\lastskip< 1 em \hskip-\lastskip
      \hskip1.0em plus0em minus0.5em \fi \nobreak
      \vrule height0.75em width0.75em depth0 em\fi}

\newcommand{\mbf}[1]{
\mathchoice{\hbox{\boldmath $\displaystyle #1$}}{\hbox{\boldmath $\textstyle #1$}}{\hbox{\boldmath $\scriptstyle #1$}}{\hbox{\boldmath $\scriptscriptstyle #1$}}}

\newcommand{\gb}{\vb{g}}
\newcommand{\hb}{\vb{h}}
\newcommand{\Gb}{\vb{G}}

\newcommand{\Xb}{\vb{X}}

\newcommand{\hU}{\hat{U}}

\newcommand{\rd}{\mathrm{d}}

\newcommand{\lambdab}{\mbf{\lambda}}

\begin{document}

\title{On the semiclassical limit of 4d spin foam models}
\author{Florian Conrady}
\email{fconrady@perimeterinstitute.ca}
\affiliation{Perimeter Institute for Theoretical Physics, Waterloo, Ontario, Canada}
\author{Laurent Freidel}
\email{lfreidel@perimeterinstitute.ca}
\affiliation{Perimeter Institute for Theoretical Physics, Waterloo, Ontario, Canada}
\preprint{PI-QG-92}

\begin{abstract}
We study the semiclassical properties of the Riemannian spin foam models with Immirzi parameter
that are constructed via coherent states. We show that, in the semiclassical limit, the quantum spin foam amplitudes of an arbitrary triangulation are exponentially suppressed if the face spins do not correspond to a discrete geometry.
When they do arise from a geometry, the amplitudes reduce to the exponential of i times the Regge action.
Remarkably, the dependence on the Immirzi parameter disappears in this limit.
 \end{abstract}

\maketitle

\section{Introduction}

Loop quantum gravity (LQG) is an approach to canonical non--perturbative quantum gravity, where the first--order (or connection) formulation of gravity plays a central role. Spin foam models arise from the attempt to construct a corresponding covariant (or path--integral) formulation of quantum gravity.
In both the canonical and covariant approach, one central open issue is the semiclassical limit---the question whether these theories reduce to general relativity in suitable semiclassical and low--energy regimes. This problem has been explored by many authors and from various angles: for example, by the use of semiclassical states \cite{ARSgravitonsloops}--\cite{Giesel}, by the extraction of propagators from spin foam models \cite{particlescattering}--\cite{AlesciRovelliI}, by numerical simulations \cite{BaezChristensenHalfordTsang,ChristensenLivineSpeziale}
and by symmetry reduction \cite{BojowaldLivingReview}--\cite{AshtekarPawloswkiSinghVandersloot}. At this stage, however, there is no conclusive evidence that LQG or spin foam models in 4 dimensions do have a satisfactory low--energy behaviour. More tangible results have been obtained in 3 dimensions, where the classical and quantum theory are far simpler \cite{FreidelLouaprePonzanoReggerevisitedI,FreidelLouaprePonzanoReggerevisitedII,FreidelLivinePonzanoReggerevisitedIII,NouiPereztthreedLQGphysicalscalarproduct}.
 In this case, spin foams were coupled to point particles \cite{FreidelLouaprePonzanoReggerevisitedI,NouiPerezpointparticles}, and it was found that the semiclassical limit is related to a field theory on non--commutative spacetime \cite{FreidelLivinePonzanoReggerevisitedIII}.

Over the last years most investigations in 4 dimensions were focused on a model that was introduced by Barrett \& Crane (BC) in 1997 
\cite{BarrettCrane}.
It can be constructed by starting from a 4d BF theory and by imposing suitable constraints on the $B$--field \cite{DePietriFreidelrelativistic,ReisenbergerLR}. These constraints are called simplicity constraints and should restrict the $B$--field such that it becomes a wedge product of two tetrad one--forms. 
This procedure for imposing simplicity was subject to various criticisms: it was argued, in particular, that the BC model could not have a realistic semiclassical limit, since its degrees of freedom are constrained too strongly.

More recently, two new techniques for constructing spin foam models  were introduced that open the way to a resolution of this difficulty: the coherent state method \cite{LivineSpezialecoherentstates}, based on integrals over coherent states on the group, and a new way of implementing the simplicity constraints  \cite{EPR1}.
 These techniques led to the definition of several new spin foam models: firstly, a model by 
Engle, Pereira \& Rovelli (EPR) \cite{EPR1,EPR2}, and later models by Freidel \& Krasnov (FK$\gamma$) \cite{FK} that incorporate any value of the Immirzi parameter $\gamma\neq 1$ and reproduce the EPR model for $\gamma=0$ \cite{FK,LS}. Engle, Pereira, Livine \& Rovelli \cite{ELPR} also studied the inclusion of the Immirzi parameter and proposed models (ELPR$\gamma$) which differ from FK$\gamma$ for $\gamma>1$. A detailed comparison of the Riemannian models has been performed in \cite{CF1}. Lorentzian versions of these models have been constructed as well \cite{FK,PereiraLorentzian,ELPR}.

In this paper, we focus our study on the set of Riemannian models FK$\gamma$. 
The main reason for this is the result of  \cite{CF1}, where we showed that each of these models can be written  as a path integral with an explicit, discrete and local action. We will use this path integral representation to analyze the semiclassical properties of the spin foam models FK$\gamma$.

As shown in \cite{CF1}, all known 4d spin foam models with gauge group SO(4) can be written in a unified manner. One first introduces a vertex amplitude $A_{v}(j^{\pm}_{f},l_{e},k_{ef})$ which depends on a choice of SO(4) representations for each face $f$ of the spin foam, a choice of SU(2) intertwiners $l_{e}$ for each edge, and a choice of SU(2) representations $k_{ef}$ for each ``wedge'' (i.e.\ each pair $(ef)$).
This vertex amplitude  is just the SO(4) 15j symbol with SO(4) representations expanded onto SU(2) ones (see \cite{CF1} for more details).
If one sums these vertex amplitudes without any constraint one simply obtains a spin foam representation of SO(4) BF theory.

The spin foam models for gravity arise from two restrictions: firstly, a restriction on the SO(4) spins in terms of the Immirzi parameter $\gamma$, namely
\be\label{j+-}
\frac{j^+}{j^-} = \frac{1+\gamma}{|1-\gamma|}\,,
\ee
which implements part of the simplicity constraints\footnote{These models are only defined for $\gamma \neq 1$. We also assume that $\gamma \ge 0$, since a change of sign $\gamma \to -\gamma$ is equivalent to swapping $j^+$ and $j^-$.}.
This implies that there is only one free SU(2) spin per face, denoted by $j_f$. 
The second  restriction pertains to the set of SU(2) ``wedge'' representations  that one should sum over. It is expressed by the choice of a non--trivial measure $D^{\gamma}_{j,k}$. The cross simplicity constraints require \cite{FK,ELPR}  that this measure should be peaked around $k= j^{+}-j^{-}$ for 
$\gamma>1$. In the ELPR$\gamma$ models, this constraint is imposed strongly, \`a la Barrett--Crane, while in the coherent state construction of FK$\gamma$ it is implemented weakly.

The partition function of the spin foam models is given by a sum over spin foams that reside on the dual of a triangulation $\Delta$ and satisfy the above constraints: 
\be
Z_{\Delta}^{\gamma} =\sum_{j_{f}}  \prod_{f}\d_{j_{f}^{\gamma+}} \rd_{j_{f}^{\gamma-}} W^\gamma_\Delta(j_f)\,,
\ee
where
\be \label{Z}
\quad  \quad 
W^\gamma_\Delta(j_f) \equiv \sum_{l_{e},k_{ef}} \prod_{e} \rd _{l_{e}} \prod_{ef} \rd_{k_{ef}}D^\gamma_{j_{f},k_{ef}}  \prod_{v} A_{v}^{\gamma}(j_{f},l_{e},k_{ef})\,.
\ee
Here, the amplitude $W^\gamma_\Delta(j_f)$ contains the sum over all intertwiner and wedge labels $l_e$ and $k_{ef}$, and can thus be regarded as an ``effective'' spin foam amplitude for given spins $j_f$. 
$\rd_{j}$ denotes the dimension of the spin $j$ representation.

The main focus of our work is to find the semiclassical asymptotics of this effective spin foam amplitude. As we will see in the next section, this amounts to determine the behaviour of $W^{\gamma}_\Delta(j_f)$ for large spins. We will find that, in this limit, the effective amplitude is exponentially suppressed if the spin labelling cannot be interpreted as areas of a discrete geometry. When the spins do arise from a discrete geometry, on the other hand, and when $\gamma > 0$, the effective amplitude $W^{\gamma}_\Delta(j_f)$ is given by the exponential of i times the Regge action. It is remarkable that the dependence on the Immirzi parameter drops out. The corresponding analysis for the EPR model yields that the exponent vanishes, i.e.\ the effective action is zero. 

The paper is organized as follows: in section \ref{pathintegralrepresentationofspinfoammodels}, we review the path integral representation that is used to derive the semiclassical approximation. In section \ref{semiclassicallimit}, we define the notion of semiclassical limit that we apply in this paper, and present the main result derived in the following sections. Section \ref{semiclassicalequations} states the equations which characterize the dominant contributions to the semiclassical limit. In \sec{rewritingtheequations} we  rewrite these equations and project them from SU(2)$\times$SU(2) to SO(4). 
In section \ref{discretegeometry}, we introduce definitions of co--tetrad, tetrad and spin connection on the discrete complex. These are needed in section \ref{discretegeometry}, where we show that the solutions to the equations are given by discrete geometries. Finally, in \sec{semiclassicalapproximationofeffectiveamplitude}, we put everything together and state the asymptotic approximation of the effective spin foam amplitude $W^{\gamma}_\Delta(j_f)$.

\section{Path integral representation of spin foam models}
\label{pathintegralrepresentationofspinfoammodels}

In this section, we review the path integral representation for the EPR and FK$\gamma$ models derived in ref.\ \cite{CF1} and introduce some notations and definitions for simplicial complexes and their duals.

In the following, $\Delta$ denotes a simplicial complex and $\Delta^*$ stands for the associated dual cell complex. We assume that $\Delta$ is \textit{orientable}.
We refer to cells of $\Delta$ as vertices $p$, edges $\ell$, triangles $t$, tetrahedra $\tau$ and 4--simplices $\sigma$. The 0--, 1-- and 2--cells of the dual complex $\Delta^*$ are called vertices $v$, edges $e$ and faces $f$ respectively. We will also need a finer complex, called ${\cal{S}}_\Delta$, which results from the intersection of the original simplicial complex $\Delta$ with the 2--skeleton of the dual complex $\Delta^*$. This leads to a subdivision of faces $f\subset\Delta^*$ into so-called wedges, and each edge $e\subset\Delta^*$ is split into two half-edges (see \fig{faceandwedge}b). We refer to oriented half--edges by giving the corresponding pair $(ve)$ or $(ev)$. When an edge in ${\cal{S}}_\Delta$ runs from the center of a face $f$ to the edge $e\subset\pa f$, it is denoted by the pair $(fe)$. A wedge is either labelled by a pair $ev$ or by the pair $ef$, where $f$ is the face that contains the wedge and $e$ is the edge adjacent to the wedge that comes \textit{first} w.r.t.\ the direction of the face orientation. 

Given ${\cal{S}}_{\Delta}$ and an orientation of its faces $f$, we define a discretized path integral that is equivalent to the spin foam sum \eq{Z}. The variables are spins $j_f$ on faces, SU(2) variables $u_e$ and $n_{ef}$ on edges and wedges respectively, and SU(2)$\times$SU(2) variables $\gb_{ve}$ and $\hb_{ef}$ on half--edges. The set of $(\gb_{ve},\hb_{ef})$ represents a discrete connection on the complex ${\cal S}_\Delta$. 
 We distinguish two types of connection variables, since there are two kinds of half--edges in ${\cal S}_\Delta$: half--edges $(ev)$ along the boundary $\pa f$ of a face $f$, and half--edges $(ef)$ that go from an edge $e$ in the boundary $\pa f$ to the center of the face $f$ (see \fig{variationinterior}).
 Given such a connection, and for a wedge orientation $[eve'f]$, we can construct the wedge holonomy
 $\Gb_{ef} = (G^+_{ef},G^-_{ef})$, where 
\be
\Gb_{ef} = \gb_{ev} \gb_{ve'} \hb_{e'\!f} \hb_{fe}\,.
\ee

\psfrag{f}{$f$}
\psfrag{e}{$e$}
\psfrag{e'}{$e'$}
\psfrag{e''}{$e^{''}$}
\psfrag{v}{$v$}
\psfrag{v'}{$v'$}
\psfrag{(a)}{(a)}
\psfrag{(b)}{(b)}
\pic{faceandwedge}{(a) Face $f$ of dual complex $\Delta^*$. (b) Subdivision of face $f$ into wedges. The arrows indicate starting point and orientation for wedge holonomies.}{3cm}{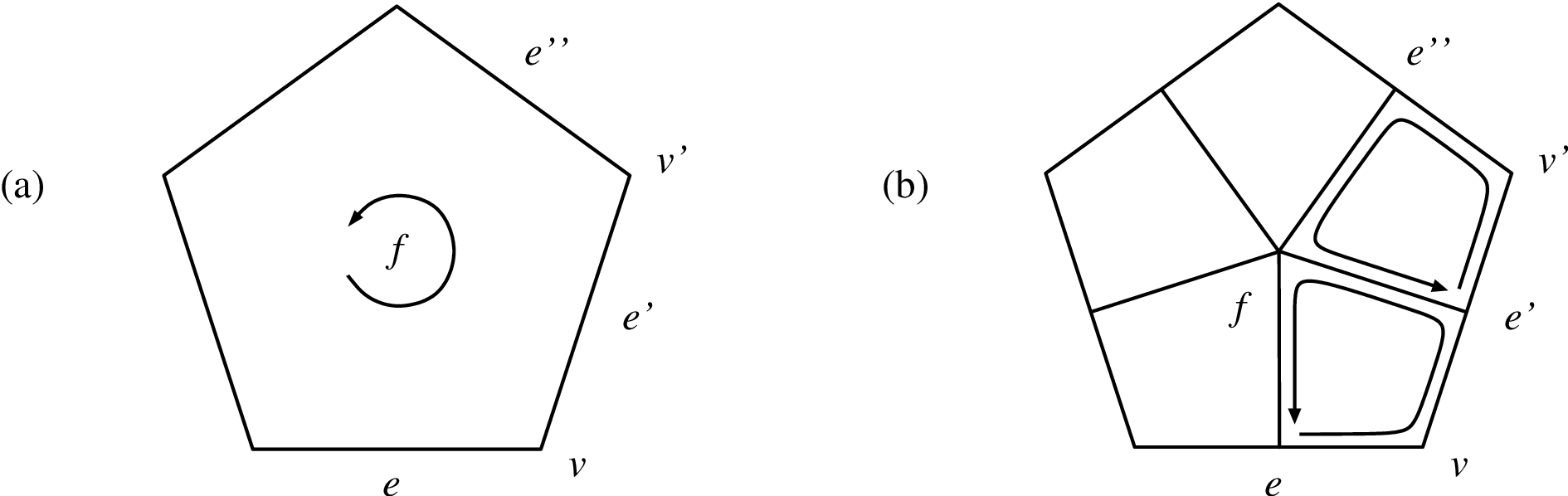}

The other set of variables $(j_{f},u_{e},n_{ef})$ represent (pre--)geometrical data\footnote{A truly geometrical interpretation is only valid on--shell, when the closure constraint is imposed.}. As we will see in more detail later, one can think of $u_{e}$ as a unit 4--vector normal to the tetrahedron dual to $e$. The spin
$j_{f}$ determines the area of the triangle dual to $f$ and $n_{ef}$ represents a vector normal to this triangle in the subspace orthogonal to $u_{e}$.
We use these variables to define Lie algebra elements $\Xb_{ef}^{\gamma} = (X^{\gamma+}_{ef}, X^{\gamma-}_{ef})\in \mathrm{su(2)}\oplus\mathrm{su(2)}$ associated with wedges of ${\cal S}_\Delta$. They depend on the value of the Immirzi parameter $\gamma$ and are given by 
\be
X^{\gamma+}_{ef} \equiv \gamma^+ j_f \,n_{ef} \sigma_3  n_{ef}^{-1}\,,\qquad  X^{\gamma-}_{ef} \equiv -\gamma^- j_f \,u_e n_{ef} \sigma_3  n_{ef}^{-1} u_e^{-1}\,.
\label{Xintermsofjun}
\ee
Here, $\sigma_i$ denotes the Pauli matrices. The spins $j_f$ are arbitrary non--negative half--integers.
$\gamma^+$ and $\gamma^-$ are the integers with smallest absolute value that satisfy $\gamma^+ > 0$ and 
\be
\frac{\gamma^+}{\gamma^-} = \frac{\gamma+1}{\gamma-1}\,.
\ee
That is, if $\gamma > 1$, both $\gamma^+$ and $\gamma^-$ are positive integers, while for $\gamma < 1$, $\gamma^-$ is 
negative\footnote{In previous papers \cite{FK,CF1}, a different convention was used, where both $\gamma^+$ and $\gamma^-$ are positive. This entails minus signs in various formulas, depending on whether $\gamma > 1$ or $\gamma < 1$. With the present convention, we no longer need to make this distinction, since the minus signs are absorbed into $\gamma^-$.}.
In the following, we sometimes use the notation 
\be
j^{\gamma\pm}_f \equiv \left|\gamma^\pm\right| j_f\,.
\ee
In the particular cases $\gamma=0$ and $\gamma=\infty$ (corresponding to the EPR and FK model), 
one recovers the usual simplicity relations \cite{BarrettCrane,FKS}, i.e.\ $j^+ = j^- = j$. 

The action of the path integral is given by
\be
S^{\gamma}_{\Delta}(j_{f},u_{e},n_{ef}; \gb_{ve},\hb_{ef}) = \sum_{e,\,f\supset e}\; 
\left(S(X_{ef}^{\gamma+}; G^{+}_{ef}) +  S(X_{ef}^{\gamma-}; G^{-}_{ef})\right)\,,
 \label{action}
\ee
where
\bea
S(X;G) &\equiv& {2|X|} \ln \tr\left[\frac{1}{2}\left(\mathbbm{1} + \frac{X}{|X|}\right)G \right].
\label{actionFK}
\eea
In the last equality, $X=X^{i}\sigma_{i}$ is a SU(2) Lie algebra element, $G$ an SU(2) group element, $|X|^{2}\equiv X^i X_i$ and the trace is in the fundamental representation of SU(2). Note that by definition $|X^{\gamma\pm}_{ef}| = j^{\gamma \pm}_f$.
This action is invariant under  gauge transformation labelled by SU(2)$\times$ SU(2) group elements $\lambdab_{e},\lambdab_{f},\lambdab_{v}$
living at vertices, faces and edges of ${\cal S}_{\Delta}$:
\be
\label{gaugesymmetry}
\gb_{ev} \to  \lambdab_{e}\gb_{ev} \lambdab_{v}^{-1}\,,\quad
\hb_{ef} \to  \lambdab_{e}\hb_{ef} \lambdab_{f}^{-1}\,,\quad
n_{ef} \to \lambda^{+}_{e}n_{ef}\,,\quad 
u_{e} \to  \lambda^{-}_{e}u_{e} (\lambda^{+}_{e})^{-1}\,.
\ee
In order to evaluate this action for a general group element $G = P_G^0\mathbbm{1} + \irm P_G$, where $P_G = P_G^i\sigma_i$ and  $P_0^2 + |P_G|^2 = 1$,
it is convenient to decompose $G$ into a part parallel to $\hat{X}\equiv X/|X|$ and a part orthogonal to it:
\be
G = \sqrt{1-\left|P_G\times \hat{X}\right|^{2}}\left(\cos \Theta \mathbbm{1} + \irm\sin \Theta \hat{X} \right) 
+ \irm \left( P_G - (P_G\cdot \hat{X}) \hat{X} \right).
\ee
where $(P_G\times \hat{X})^i \equiv \epsilon^i{}_{jk} P_G^j \hat{X}^k$ and $\cos\Theta = P^0_G / \sqrt{1 - |P_G\times \hat{X}|^2}$.
The action is in general complex, since
\be
S(X;G) = |X| \ln\left(1-\left|P_{G}\times \hat{X}\right|^{2}\right)
+ 2\irm |X| \Theta\,.
\ee
It is important to note that the real part of this action is always negative; $\mathrm{Re}(S(X,G))\leq 0$. 
It is zero only if $\hat{X}$ is parallel to $P_{G}$ or equivalently if the Lie algebra element $X$commute with the group element $G$. In this case the action is purely imaginary and has the ``Regge'' form $S(X;G) = 2\irm |X| \Theta$.

As shown in \cite{CF1}, the spin foam models  FK$\gamma$ introduced in \cite{FK} and described in \eq{Z} can be written as
\be
Z^{\gamma}_{\Delta}(j_f)=\sum_{j_{f}}\prod_{f} \rd_{j_{f}^{\gamma+}}\rd_{j_{f}^{\gamma-}} W_{\Delta}^{\gamma}(j_{f}),
\ee
where the effective amplitude $W_{\Delta}^{\gamma}$ is obtained by integration over  all the variables\footnote
{\label{udep}  Note that thanks to the gauge symmetry described in \eq{gaugesymmetry} there is no need to integrate over the variables $u_e$. The effective amplitude obtained after integration over all variables except $j_f$ and $u_e$ is independent of $u_e$.} except $j_{f}$:
\be\label{Zeff}
W_{\Delta}^{\gamma}(j_{f}) = 
\int\prod_e \d u_e \prod_{e,\,f\supset e} \rd_{j_f}^{\gamma+}\rd_{j_f}^{\gamma-} \d n_{ef}\int  \displaystyle
\prod_{v,\,e\supset v} \d {\bf g}_{ev}   \prod_{e,\,f\supset e}\d {\bf h}_{ef}\; e^{S^{\gamma}_{\Delta}(j_{f},u_{e},n_{ef}; \gb_{ve},\hb_{ef})}.
\ee

\section{Semiclassical limit}
\label{semiclassicallimit}

In this section, we define the notion of semiclassical limit that we investigate in this paper, and state our main results.
We focus our interest on the effective amplitude $W_{\Delta}^{\gamma}(j_{f})$, which depends only on the scalars $j_{f}$ associated to each face.
In order to define a semiclassical limit we need to reinstate the $\hbar$--dependence and introduce dimensionful quantities. The spins $j_{f}$ are then proportional to the physical area.

The Immirzi parameter enters in the relationship between the discrete bivector field $X_{ef}^{\gamma IJ}$ and the dimensionful simple area bivector field $A_{ef}^{IJ}$ associated with the triangle dual to $f$: namely,\footnote{The map between bivectors $X_{ef}^{IJ}$ and Lie algebra elements $\Xb_{ef} = (X^{+ i}\sigma_i,X^{- i}\sigma_i)$ of $\mathrm{su(2)}\oplus\mathrm{su(2)}$ is given by $X^{\pm i}_{ef} = \frac{1}{2}\,\epsilon^i{}_{jk} X^{jk}_{ef} \pm X^{0i}_{ef}$.}
\be
\label{XA}
(16\pi \hbar  G) X^\gamma_{ef}  = \star  A_{ef} + \frac{1}{\gamma} A_{ef}\,,
\ee
when $ \gamma > 0$.

The simplicity of the area bivector implies that $|A_{ef}^{+}| =|A_{ef}^{-}| \equiv  \clA_{f}$, where  $\clA_{f}$  denotes the physical area of  the triangle dual to $f$. The relationship \eq{XA} can be written $(16\pi \hbar  G)\gamma X^{\gamma\pm}_{ef} =  (1\pm\gamma) A_{ef}^{\pm}$, which leads to 
\be
\frac{\clA_{f}}{8\pi \hbar G} = \left(\gamma^+ + \gamma^-\right) j_f\,.
\label{areaintermsofspins}
\ee 
 
We implement the semiclassical limit by taking $\hbar$ to zero, while keeping the physical dimensionful areas $\clA_{f}$ fixed. The previous equation \eq{areaintermsofspins} tells us that in this limit the spins $j_f$ are uniformly rescaled to infinity. Thus, the semiclassical regime is reached by taking the limit $N \to \infty$ of the amplitude $W^{\gamma}_\Delta(Nj_{f})$ in \eq{Zeff}.

Since the action is linear in $j_f$, this corresponds to a global rescaling of the action by $N$. Hence the limit $N\to\infty$ is controlled by the  stationary phase points of the exponent: the integral localizes as a sum over contributions from stationary phase points.
Moreover, as we have seen, the action is complex with a negative real part. As a result, stationary phase points which do not lie at the maximum $\Re(S_{\Delta}^{\gamma})=0$ are \textit{exponentially} suppressed. Altogether this means that the semiclassical limit is controlled by \textit{stationary points} of $S_{\Delta}^{\gamma}$ which are also \textit{maxima} of the real part $\Re(S_{\Delta}^{\gamma})$. A more detailed discussion of the asymptotic analysis is given in \sec{semiclassicalapproximationofeffectiveamplitude}.  
Before stating our main result, we have to recall that in the continuum the equivalence between gravity and the constrained BF formulation is only established if one imposes a condition of non--degeneracy on the $B$ field\footnote{see \cite{DePietriFreidelrelativistic, ReisenbergerLR} for a more detailed discussion of this point and the potential problems due to degenerate configurations in the path integral.}. We therefore need to distinguish between non--degenerate and degenerate configurations in our analysis. This is achieved by splitting the amplitude \eq{Zeff} into two parts\footnote{see \cite{FLouapre6j} for an analysis of stationary points of group integrals representing the $6j$-- and $10j$--symbol using a similar splitting.},
\be
W^{\gamma}_\Delta(j_{f}) = W^{ \mathrm{ND}\gamma}_{\Delta }(j_{f}) + W^{\mathrm{D}\gamma}_{\Delta }(j_{f})\,,
\ee
where $W^{\mathrm{ND}\gamma}_{\Delta }(j_{f})$ is defined by the integral \eq{Zeff} subject to the constraint that
\be 
\label{nondeg}
\left|\epsilon_{IJKL} X_{ef}^{IJ} (g_{ee'}\triangleright X_{e'f'})^{KL}\right| \quad > \quad 0
\ee 
for all pairs of wedges $(ef)$ and $(e'f')$ that share a vertex, but do not share an edge. Here, 
$g_{ee'}\triangleright X_{e'f'} \equiv g_{ee'}\triangleright X_{e'f'} g^{-1}_{ee'}$ and  $g_{ee'}\equiv g_{ev}g_{ve'}$. 
The term $W^{\mathrm{D}\gamma}_{\Delta }(j_{f})$ denotes the complementary integral consisting of degenerate configurations.

One of the characteristics of 4d spin foam models is the assignment of spins $j_f$ to each face $f$ of the dual complex $\Delta^*$ and of corresponding areas $\clA_t(j_f)$ to each triangle $t$ of $\Delta$. In contrast, Regge calculus is based on an assignment of a discrete metric to the complex, defined by lengths $l_\ell$ associated with each edge $\ell\subset\Delta$ and subject to triangle inequalities. The areas $\clA_t$ of triangles $t$ dual to faces $f$ are then determined as a function $\clA_t(l_\ell)$ of the edge lengths $l_\ell$. It is well--known \cite{barrett, Makela, dittrichspeziale} that for an arbitrary assignment of spins $j_f$, there is, in general, no set of $l_\ell$'s such that $\kappa j_f= \clA_t(l_\ell)$. The set of areas $\clA_t$ determines at least one flat geometry inside each 4--simplex, but the geometries of tetrahedra generally differ, when viewed from different 4--simplices. In the following, we will call an assignment of spins $j_f$ Regge--like if there is a discrete metric $l_\ell$, $\ell\subset\Delta$, such that $\clA_t(j_f) = \clA_t(l_\ell)$.

Our principal result is that the set of stationary points of the integral \eq{Zeff} which are non--degenerate and have a maximal real part, are Regge--like. Moreover, the on-shell action is exactly the Regge action! This result relies on the specific realization of the spin foam model in terms of the local action \eq{action} which is valid for the FK$\gamma$ version \cite{FK} of the model. It does not apply to the ELPR$\gamma$ construction \cite{ELPR}, which is different from FK$\gamma$ for $\gamma>1$ (see \cite{CF1} for a comparison).

More precisely, a configuration $(j_f, u_e, n_{ef},\gb_{ev},\hb_{fe})$ is a solution of the conditions
\be
 \pdiff{S}{n_{ef}} = \pdiff{S}{u_e} = \pdiff{S}{g_{ev}} = \pdiff{S}{h_{ef}} = 0\,,\qquad \Re S = 0\,,
\label{stationaryandmaximal} 
\ee
and eq.\ \eq{nondeg}, if and only if the spins $j_f$, $f\subset\Delta^*$, are Regge--like. 
In this case, there exist edge lengths $l_\ell$, $\ell\subset\Delta$, such that $j_f = (\gamma^+ + \gamma^-)^{-1} \clA_f(l_\ell)$ for $f = t^*$.
Moreover, for such a solution we have, as long as $\gamma\neq 0$, that
\be
S_\Delta^\gamma = \sum_f \clA_f(l_\ell) \Theta_f(l_\ell) \equiv S_R(l_\ell)\,,
\ee
where $\Theta_f(l_\ell) $ is the deficit angle associated with the face $f$ and $\clA_f(l_\ell)$ is the area in Planck units. 
If $\gamma=0$, the on-shell action vanishes, i.e.\ $S_{\Delta}^{0}= 0$, in agreement with the fact that $\gamma=0$ corresponds to a topological theory classically (see \cite{FK}).

It is important to note that the dependance on $\gamma$ has \textit{disappeared} from the functional form of the action.
This parallels the behaviour of the continuum theory, where the $\gamma$ dependence drops out classically, once we solve the torsion equation.
It also provides a non--trivial check on whether the chosen spin foam model captures the right semiclassical dynamics. The dependence on the Immirzi parameter arises only at the quantum level as a quantization condition on the area\footnote{The dimensionful area has to satisfy the condition that $\clA_t (8\pi\hbar G)^{-1} (\gamma^+ + \gamma^-)^{-1}$ is a half--integer. This quantization condition becomes invisible in the semiclassical limit $\hbar\to\infty$.}, similar as in canonical loop quantum gravity.

These results are derived in section \ref{solutions} and \ref{semiclassicalapproximationofeffectiveamplitude}, and imply the following statements on the effective amplitude $W^{\mathrm{ND}\gamma}_{\Delta }(j_f)$: as $N\to\infty$, the amplitude $W^{\mathrm{ND}\gamma}_{\Delta }(Nj_f)$ is exponentially suppressed\footnote{That is, the limit $N\to \infty$ of  $N^n W^{\mathrm{ND}\gamma}_{\Delta }(Nj_f)$ is equal to zero for all $n\in\bN$.}, if the spins $j_f$, $f\subset\Delta^*$, do not arise from a Regge geometry. On the other hand, if the $j_f$'s are Regge--like, there is a non--zero function $c_\Delta(j_f)$, independent of $N$, such that
\be
W^{\mathrm{ND}\gamma}_{\Delta }(Nj_f) \sim \frac{c_\Delta(j_f)}{N^{\frac{r_\Delta}2}} \left(\exp\left(\irm N S_R\right) + \mbox{c.c}\right)
\label{asymp}
\ee
as $N\to\infty$. Here, $c.c$ stands for the complex conjugate.
The number $r_\Delta$ is the rank of the Hessian and given by
\be
r_{\Delta} = 33 E - 6V - 4F\,,
\ee
with $V$, $E$ and $F$ denoting the number of vertices, edges and faces of $\Delta^*$.

This shows that, in the semiclassical limit, the effective amplitude $W^{ND\gamma}_{\Delta}(j_f)$ is described by an effective action, which is the Regge action.
If there are several discrete geometries $l_\ell$, $\ell\subset\Delta$, for a given set $j_f$, $f\subset\Delta^*$, one should sum over them in the asymptotic evaluation \eq{asymp}. In the following sections, we prove the above statements and study in detail the non--degenerate solutions to eqns.\ \eq{stationaryandmaximal}.

\section{Classical equations}
\label{semiclassicalequations}

We will now derive the explicit form of the equations that follow from the conditions $\delta S = 0$ and $\Re S = 0$. 

\subsection{Variation on interior and exterior edges}

\psfrag{e1}{$e$}
\psfrag{e2}{$e'$}
\psfrag{e3}{$e''$}
\psfrag{v1}{$v$}
\psfrag{v2}{$v'$}
\psfrag{v3}{$v''$}
\pic{variationinterior}{Variation of the group variable $\hb_{fe}$ on the edge $fe$ in the interior of the face $f$.}{3cm}{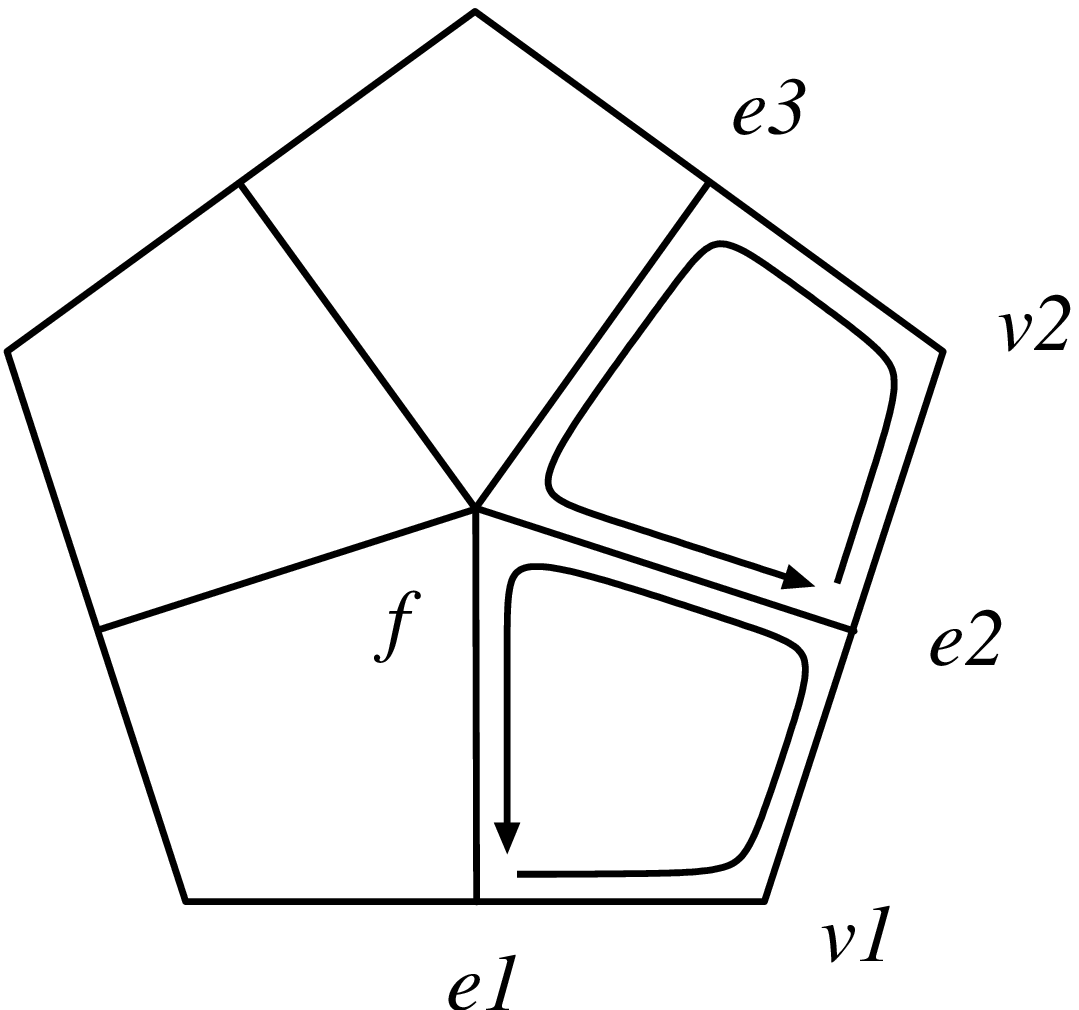}

We first consider the variation of the variable $h^\pm_{e'f}$.
Since the edge $e'f$ belongs to two wedges, denoted $ef$ and $e'f$, the variation of the action involves only two terms. If $e$ is the edge preceding $e'$ along the orientation of the face $f$,
one has $\Gb_{ef}=\gb_{ev}\gb_{ve'}\hb_{e'f}\hb_{fe}$ and 
$\Gb_{e'f}=\gb_{e'v'}\gb_{v'e''}\hb_{e''f}\hb_{fe'}$.
We can write the variation of the action as (see \fig{variationinterior}) 
\bea
\delta S 
&=& 
2j^{\gamma\pm}_f\tr\left[
\left(
\frac{h^\pm_{e'e} (\mathbbm{1} + \hat{X}^\pm_{ef}) G^\pm_{ef} \left(h^\pm_{e'e}\right)^{-1}}{\tr\left((\mathbbm{1} + \hat{X}^\pm_{ef}) G^\pm_{ef}\right)}\, 
-
\frac{(\mathbbm{1} + \hat{X}^\pm_{e'f}) G^\pm_{e'f}}{\tr\left((\mathbbm{1} + \hat{X}^\pm_{e'f}) G^\pm_{e'f}\right)}\, 
\right)  \delta h^\pm_{e'f}h^\pm_{fe'}
\right] =0\,,
\label{variationinterioredge}
\eea
where we use the abbreviations \setlength{\jot}{0.4cm}
\be
g^\pm_{ee'} \equiv g^\pm_{ev}g^\pm_{ve'}\,, \quad
h^\pm_{ee'} \equiv h^\pm_{ef}h^\pm_{fe'},\,\quad 
(h^{\pm}_{ef})^{-1}=h_{fe}^{\pm}, \quad (g^{\pm}_{ev})^{-1}=g_{ve}^{\pm}\,,
\ee \setlength{\jot}{0cm}
and $\hat{X}^\pm_{ef} \equiv  {X}^{\gamma\pm}_{ef}/ j^{\gamma\pm}$, which is independent of $\gamma$.

To write these equations in a more compact manner, let us define the matrix element 
\be
\hat{Y}^\pm_{ef}\equiv \frac{2(\mathbbm{1} + \hat{X}^{\pm}_{ef})}{\tr\left((\mathbbm{1} + \hat{X}^{\pm}_{ef}) G^\pm_{ef}\right)}.
\ee
Since $\delta h h^{-1}$ is in the Lie algebra, we conclude from \eq{variationinterioredge} that the traceless part of the expression in round brackets has to be zero. Moreover, since $\tr(Y_{ef}^{\pm}G_{ef}^{\pm}) = 2$, one simply gets 
\be
\fbox{\parbox{8cm}{$$
h^\pm_{fe} \left(\hat{Y}^\pm_{ef}G^{\pm}_{ef}\right) h^\pm_{ef} 
= 
h^\pm_{fe'} \left(\hat{Y}^\pm_{e'f}G^{\pm}_{e'f}\right) h^\pm_{e'f}\,.
$$}}\label{intclos}
\ee
We refer to this equation as the \textit{interior closure constraint}, since it encodes a relation between wedges in the interior of the face $f$.

\psfrag{e}{$e$}
\psfrag{v}{$v$}
\psfrag{ei}{$e_i$}
\psfrag{ej}{$e_j$}
\psfrag{fi}{$f_i$}
\psfrag{fj}{$f_j$}
\pic{variationexterior}{Variation of the group variable $\gb_{ev}$ on the segment $ev$ between faces.}{5cm}{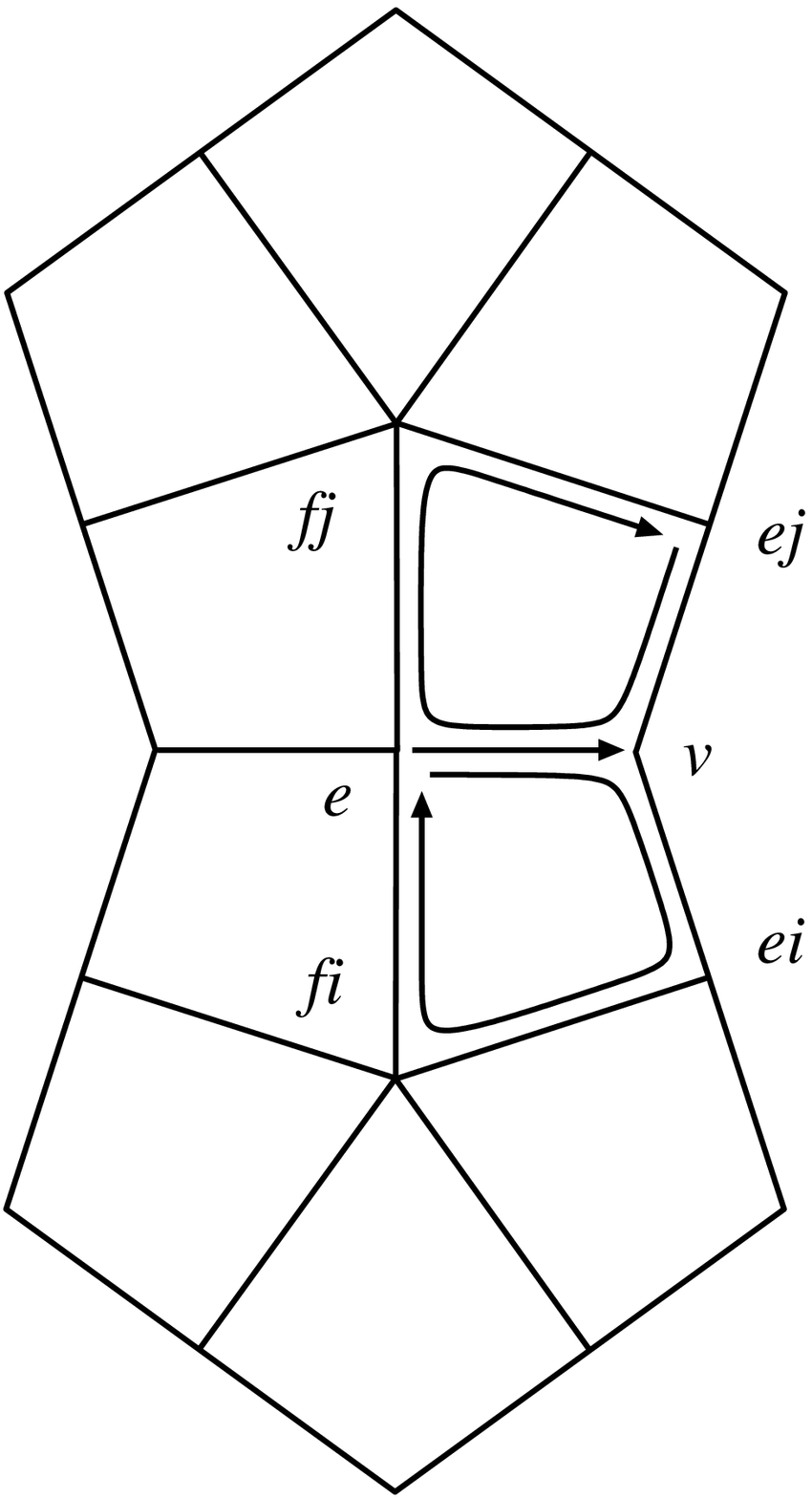}

\setlength{\jot}{0.4cm}
Next, we vary a group variable $\gb_{ev}$ on a half--edge $ev$. This calculation is slightly more involved, since the orientation of different faces has to be taken into account.  At the edge $e$, four faces $f_i$, $i = 1,\ldots,4$, intersect. Let $I^+_e$ be the set of indices $i$ for which the orientation of $f_i$ is ``ingoing'' at the vertex $v$, i.e.\ parallel to the orientation of the half--edge $(ev)$. In these cases, the wedge holonomy has the form $G_{ef_i}= g_{ev}g_{ve_{i}}h_{e_if_i}h_{f_ie}$. Denote by $I^-_e$ the complementary set for which the holonomy is $G_{e_if_i}= g_{e_iv}g_{ve}h_{ef_i}h_{f_ie_i}$. Then, variation of $g^\pm_{ev}$ gives
\bea
\delta S 
&=&
\tr\left[
\left(
\sum_{i\in I^+_e}{j^\pm_{f_i}}\,  G^\pm_{ef_i}\hat{Y}^\pm_{ef_i}
-
\sum_{j\in I^-_e} {j^\pm_{f_j}}\, g^\pm_{ee_j} G^\pm_{e_jf_j}\hat{Y}^\pm_{e_jf_j} \left(g^\pm_{ee_j}\right)^{-1}
\right) \delta g^{\pm}_{ev} g^{\pm}_{ve}
\right]= 0
\eea
Again, the traceless part of the quantity in round brackets has to be zero. Therefore,
\be
\fbox{\parbox{12.5cm}{$$
\sum_{i\in I^+_e} {j^\pm_{f_i} \left(G^\pm_{ef_i}\hat{Y}^\pm_{ef_i} - \mathbbm{1}\right)}-
\sum_{j\in I^-_e}{j^\pm_{f_j}\, g^\pm_{ee_j} \left(G^\pm_{e_jf_j}\hat{Y}^\pm_{e_jf_j} - \mathbbm{1}\right)} \left(g^\pm_{ee_j}\right)^{-1} 
= 0\,.
$$}}\label{extclos}
\ee
This equation relates wedges from different faces, so we call it the \textit{exterior closure constraint}.

\subsection{Variation of $u_e$ and $n_{ef}$ and maximality}

For the variation w.r.t.\ $n_{ef}$, we use the definition \eq{Xintermsofjun} of $X^\pm_{ef}$ and get
\be
\delta X^{+}_{ef} = \left[\delta n_{ef} n_{ef}^{-1}, X^{+}_{ef}\right]\,,\quad\quad
\delta X^{-}_{ef} = \left[ u_{e}\delta n_{ef} n_{ef}^{-1} u_{e}^{-1}, X^{-}_{ef}\right]\,.
\ee
The variational equation for $n_{ef}$ is therefore given by
\be
\label{neq}
\left[\hat{Y}^+_{ef}, G^+_{ef}\right] +
 u_e^{-1}\left[\hat{Y}^-_{ef}, G^-_{ef}\right] u_e= 0\,.
\ee
Similarly, by varying $u_e$ one obtains
\be
\label{ueq}
\sum_{f\supset e} \left[\hat{Y}^-_{ef}, G^-_{ef}\right]\ = 0\,.
\ee

The action being complex, its stationarity is not enough to determine the dominant contribution to the semiclassical limit. One also has to demand that the stationary points are a maximum of the real part of the action.
Since
\be
\mathrm{Re}(S_{\Delta}^{\gamma})= 
\sum_{(ef)} \left[j^{\gamma+}_{f} \ln\left(1-\frac14 \left|\left[\hat{X}^{+}_{ef},G^{+}_{ef}\right]\right|^2\right)
+ j^{\gamma-}_{f} \ln\left(1-\frac14 \left|\left[\hat{X}^{-}_{ef},G^{-}_{ef}\right]\right|^2\right)\right]\,,
\ee
and $\left|\left[X,G\right]\right|^2\geq 0$, this is  maximal when
\be\label{max}
\fbox{\parbox{6cm}{$$\left[G^{+}_{ef},\hat{X}^{+}_{ef}\right]=0=\left[G^{-}_{ef},\hat{X}^{-}_{ef}\right]\,.
$$}}
\ee
Note that the maximum condition implies that the stationarity equations (\ref{neq},\ref{ueq}) for $n_{ef}$ and $u_{e}$ are automatically fulfilled.
Moreover it is important to notice that this relation leads to a 
drastic simplification of the closure constraints, since it leads to the identity:
\be
\label{YGX}
\hat{Y}^{\pm}_{ef}G_{ef}^{\pm} = G_{ef}^{\pm}\hat{Y}^{\pm}_{ef} = \mathbbm{1} + \hat{X}^{\pm}_{ef}\,.
\ee

\section{Rewriting the equations}
\label{rewritingtheequations}

\subsection{Parallel transport to vertices}

To analyze the variational equations, it is convenient to make a change of variables. The original variables 
$\Xb_{ef}^{\gamma}, \Gb_{ef}$ are based at the edge $e$, which means that under gauge transformation they transform as 
$(\Xb_{ef}, \Gb_{ef}) \to (\lambdab_{e}\Xb_{ef}^{\gamma}\lambdab_{e}^{-1}, \lambdab_{e}\Gb_{ef}\lambdab_{e}^{-1})$.
The new variables are based at $v$ and defined by parallel transporting the original variables to the nearest vertices of the dual complex $\Delta^*$:  
\bea
\Xb_{ef}^{\gamma}(v) \equiv \gb_{ve} \Xb_{ef}^{\gamma} \gb^{-1}_{ve}\,, 
\label{paralleltransport}\quad\quad\Gb_{ef}(v) \equiv \gb_{ve} \Gb_{ef} \gb^{-1}_{ve}\,, \quad \quad
u_e(v) \equiv g^-_{ve} u_e \! \left(g^+_{ve}\right)^{-1}\,.
\eea
Since every edge $e$ intersects with two vertices $v$ and $v'$, this leads to a doubling of the number of variables.
This is compensated by equations that relate variables at neighbouring vertices $v$ and $v'$: i.e.\
\be
\Xb_{ef}^{\gamma}(v') = \gb_{v'v} \Xb_{ef}^{\gamma}(v) \left(\gb_{v'v}\right)^{-1}\,,\quad\quad 
\quad \quad u_e(v') = g^-_{v'v} u_e(v) \! \left(g^+_{v'v}\right)^{-1}\,.
\label{relationbetweenvertices}
\ee 
In terms of the new variables, the interior closure constraint \eq{intclos} becomes 
\be
h^\pm_{fe} g^\pm_{ev}\,
{\left(\hat{Y}^\pm_{ef}(v) G^\pm_{ef}(v)\right)}\left(h^\pm_{fe} g^\pm_{ev}\right)^{-1}
= 
h^\pm_{fe'}g^\pm_{e'v}{\left(\hat{Y}^\pm_{e'f}(v) G^\pm_{e'f}(v)\right)} \left(h^\pm_{fe'}g^\pm_{e'v}\right)^{-1}\,.
\ee
Thus,
\be
{G^\pm_{ef}(v) \hat{Y}^\pm_{ef}(v)}= 
{\hat{Y}^\pm_{e'f}(v) G^\pm_{e'f}(v)},\label{interiorclosureconstraint}
\ee
where the edge $e'$ follows the edge $e$ in the orientation of $f$.
Likewise, after conjugation by $g_{ve}^{\pm}$, the exterior closure constraint takes the form
\be
\sum_{i\in I^+_e} {j^\pm_{f_i} \left(G^\pm_{ef_i}(v) \hat{Y}^\pm_{ef_i}(v) - \mathbbm{1}\right)}
-
\sum_{j\in I^-_e} {j^\pm_{f_j} \left(G^\pm_{e_jf_j}(v) \hat{Y}^\pm_{e_jf_j}(v) - \mathbbm{1}\right)}
= 0\,.
\label{exteriorclosureconstraintY}
\ee \renewcommand{\arraystretch}{3}
If we impose, in addition, the maximality constraint \eq{YGX}, the closure constraints simplify and we remain with the following set of equations:
\be\label{eqsd}
\fbox{\parbox{15 cm}{
\[
\begin{array}{l@{\quad}l@{\quad}l}
\left[\Gb_{ef}(v),\Xb_{ef}^{\gamma}(v)\right] = 0, 
& 
\Xb^{\gamma}_{ef}(v') = \gb_{v'v} \Xb^{\gamma}_{ef}(v) \left(\gb_{v'v}\right)^{-1},
&
u_e(v') = g^-_{v'v} u_e(v) \! \left(g^+_{v'v}\right)^{-1}, 
\\ 
\Xb_{ef}^{\gamma}(v) =  \Xb_{e'f}^{\gamma}(v)\,, & \ds\sum_{f\supset e} \epsilon_{ef}(v) \Xb_{ef}^{\gamma}(v) = 0. & \\
\end{array}
\]
}}
\ee
$\epsilon_{ef}(v)$ is a sign factor which is $1$ when $f$ is ingoing at $v$, i.e.\ oriented consistently with the half edge $(ev)$, and $-1$ otherwise.
These equations are supplemented by the simplicity constraints
\be
\Xb^{\gamma}_{ef}(v) = \left(\gamma^{+} j_{f}\mathrm{Ad}(n_{ef}(v))\sigma_3\,,\, -\gamma^{-} j_{f}\mathrm{Ad}(u_{e}(v)n_{ef}(v))\sigma_3 \right)\,,
\label{Xintermsofjunagain}
\ee
where
\be
n_{ef}(v) \equiv g^+_{ve} n_{ef} \!\left(g^+_{ve}\right)^{-1}\,.
\ee
\renewcommand{\arraystretch}{1}

\subsection{Projection to SO(4)}

In order to solve these equations explicitly it is convenient to project them  from SU(2)$\times$SU(2) to SO(4) and work purely in terms of 
vectorial and SO(4) variables.

The action has the property that
\be
S^{\gamma}_{\Delta}(j_{f},u_{e},n_{ef}; -\Gb_{ef})
= S^{\gamma}_{\Delta}(j_{f},u_{e},n_{ef}; \Gb_{ef})+ 2i \pi (\gamma^{+}+\gamma^{-})j_{f}\,,
\ee
so the weight $\exp(S^{\gamma}_{\Delta})$ projects down to a function of SO(4) if one restricts to configurations for which 
$(\gamma^{+}+\gamma^{-})j_{f}$ is an integer. We assume from now on that this is the case. 

The projection to SO(4) means that we work with bivectors $X^{IJ}$ instead of pairs $(X^{+},X^{-})$, the relation between the two being 
\be
X^{\pm}_i = \frac{1}{2}\,\epsilon_{i}{}^{jk} X_{jk} \pm X_{0i}\,.
\label{3vectorsintermsofbivectors}
\ee
We also associate a unit vector $\hU_e$ in $\bR^4$ to each SU(2) element $u_e$, defined by the relation
\be
u_e = \hU^0_e \mathbbm{1} + \irm\hU^i_e\,\sigma_i\,,\quad \hU_e^2=1\,,
\label{relationuUhat}
\ee
where $\sigma_i$ are the Pauli matrices. 
To translate the simplicity constraints \eq{Xintermsofjunagain} to so(4), it is convenient to introduce a fiducial bivector field which is independent of $\gamma$ and which is simple unlike $X^{\gamma}$.
We denote this  bivector field by  $X_{ef}$ without any subscript $\gamma$ and it is defined by 
\be
X^{\gamma +}_{ef} = \gamma^+ X^+_{ef}\,,\qquad X^{\gamma -}_{ef} = -\gamma^- X^-_{ef }\,.
\label{XgammaX}
\ee
Due to the simplicity constraint \eq{Xintermsofjunagain}, $u_{e}X^{+}_{ef} + X^{-}_{ef}u_{e} =0$ and $|X_{ef}^{\pm}|=j_{f}$. By using the identity
\be
\frac{1}{2\irm}\left( uX^{+}+X^{-}u \right)
=(\star  X\cdot \hU)_{0} + \irm (\star  X\cdot \hU)_{i}\sigma^{i},\quad (X\cdot U)_{I}\equiv X_{IJ}U^{J},\quad  (\star  X)_{IJ}\equiv\frac12 \epsilon_{IJKL}X^{KL}.
\ee
we then find that
\be 
\label{simp1}
\fbox{\parbox{13cm}{$$
X^{\gamma}_{ef} = \frac{1}{2}(\gamma^{+}+ \gamma^{-}) X_{ef}
+ \frac{1}{2}(\gamma^{+}- \gamma^{-})(\star   X_{ef}),\quad \quad  (\star  X_{ef}(v)\cdot \hU_{e}(v))^{I} = 0
$$}}
\ee 
with $X_{ef}(v)\cdot X_{ef}(v) = 2 j_{f}^{2}$. This equation is the discrete version of the simplicity constraints in the continuum.
Recalling the definition of $\gamma$ in terms of $\gamma^{\pm}$, we can write the relation between 
$X^{\gamma} $ and $X$ also as 
\bea
X^\gamma_{ef} = \frac{1}{2}\left(\gamma^+ + \gamma^-\right)\left(\star X_{ef} + \frac{1}{\gamma} X_{ef}\right)\,,
\label{XX}
\eea
which shows that for $\gamma > 0$ $X$ plays the role of the dual of the area bivector:
\be
X_{ef } = \frac{1}{\gamma^{+}+\gamma^{-}}\frac{\star   A_{ef}}{8\pi\hbar G}\,.
\ee

\renewcommand{\arraystretch}{3}
Together with these simplicity conditions, we want to solve the equations \eq{eqsd}.
When written in terms of the $\gamma$--independent, simple bivector $X_{ef}$, they take the form
\be
\label{eqvect}
\fbox{\parbox{14cm}{
\[
\begin{array}{l@{\quad}l@{\quad}l}
G_{ef}(v)\triangleright X_{ef}(v) = X_{ef}(v)\,, 
& 
X_{ef}(v') = g_{v'v} \triangleright X_{ef}(v)\,,
&
\hU_{e}(v')= g_{v'v}\hU_{e}(v)\,, 
\\
X_{ef}(v) =  X_{e'f}(v)\,, & \ds\sum_{f\supset e} \epsilon_{ef}(v) X_{ef}(v) = 0\,. & \\
\end{array}
\]
}}
\ee
$\triangleright$ denotes the action of SO(4) generators on bivectors. This and eq.\ \eq{simp1} are the final form of the equations that we will study now.
\renewcommand{\arraystretch}{1}

\section{Discrete geometry}
\label{discretegeometry}

In order to find the general solution, we will assume that the bivectors $X_f(v)$ are non--degenerate: that is, 
\be
\label{nondeg2} 
X_{ef}(v)\wedge X_{e'f'}(v)\neq 0
\ee 
for any pair of faces $f,f'$ which do not share an edge. It turns out  that the solutions exist only if the set $(j_f)_f$ is Regge--like. That is, only if there is a discrete metric on the triangulation $\Delta$ for which $j_f$ is  the area of triangles dual to $f$. As we will see, the unit vectors $\hU^{I}$  are, on-shell, the normalized tetrad vectors associated with this metric and the connection $g_{v'v}$ is the discrete spin connection for this tetrad. In order to demonstrate these statements, we first need to define all these notions on the discrete complex\footnote{For previous definitions in the literature, see e.g.\ \cite{Frohlich,Caselleetc,Gionti}}.

\subsection{Co--tetrads and tetrads on a simplicial complex}
\label{cotetradandtetradsonasimplicialcomplex}

\begin{definition}
A co--tetrad $E$ on the simplicial complex $\Delta$ is an assignment of vectors $E_\ell(v)\in\bR^4$ to each vertex $v\subset\Delta^*$ and
oriented edge $\ell\subset\Delta$, $\ell\subset \sigma = v^*$, where the following properties hold:
\begin{itemize}
\item[(i)] $E_{-\ell} = -E_\ell$.
\item[(ii)] For any triangle $t\subset \sigma = v^*$, and edges $\ell_1, \ell_2, \ell_3\subset t$ s.t.\ $\partial t = \ell_1 + \ell_2 + \ell_3 $, the vectors $E_\ell(v)$ close, i.e.\ 
\be
E_{\ell_1}(v) + E_{\ell_2}(v) + E_{\ell_3}(v) = 0\,.
\label{cotetradclosure}
\ee
\item[(iii)] 
For every edge $e = v'v \subset\Delta^*$ and for any pair of edges $\ell_1$ and $\ell_2$ in the tetrahedron $\tau$ dual to $e$, we have
\be
E_{\ell_1}(v')\cdot E_{\ell_2}(v') = E_{\ell_1}(v)\cdot E_{\ell_2}(v)\,.
\label{cotetradmetricity}
\ee
\end{itemize}
\end{definition}
In other words, a co--tetrad $E$ is an assignment of a closed $\bR^4$--valued 1--chain $E(v)$ to each 4--simplex $\sigma = v^*$ that fulfills a compatibility criterion. In each 4--simplex $\sigma^* = v$, the co--tetrad vectors $E_l(v)$ define a flat Riemannian metric $g_v$ by
\be
g_{\ell_1\ell_2}(v) = E_{\ell_1}(v)\cdot E_{\ell_2}(v)\,,\quad \ell_1,\ell_2\subset\sigma\,.
\ee
Condition (iii) requires that, for any pair of 4--simplices $\sigma = v^*$ and $\sigma' = v'^*$ which share a tetrahedron $\tau$, the metric induced on $\tau$ by $E(v)$ and $E(v')$ are the same. Thus, the co--tetrad $E$ equips $\Delta$ with the structure of a piecewise flat Riemannian simplicial complex. 

We call a co--tetrad $E$ \textit{non--degenerate} if  at every vertex $v = \sigma^*\subset\Delta^*$ and for every tetrahedron $\tau\subset\sigma$, the span of the vectors $E_\ell(v)$, $\ell \subset\tau$, is 4--dimensional. 

\begin{proposition}
Given any non--degenerate co--tetrad $E$, there is a unique SO(4) connection $\Omega$ on $\Delta^*$ that satisfies the condition
\be
E_\ell(v') = \Omega_{v'v} E_\ell(v)\qquad \forall\; v'v = e\subset \Delta^*,\; \ell\subset e^*\,.
\label{conditionLeviCivitaconnection}
\ee 
We call this connection $\Omega$ the spin connection associated to $E$. 
\label{propositionLeviCivitaconnection}
\end{proposition}
\begin{proof}
Let $\hU(v)$ and $\hU(v')$ denote unit normal vectors to $E_{\ell_i}(v)$, $i = 1,2,3$, and $E_{\ell_i}(v')$, $i = 1,2,3$, respectively. 
Choose these unit normal vectors such that
\be
\sgn\det\left(E_{\ell_1}(v'),E_{\ell_2}(v'),E_{\ell_3}(v'),\hU(v')\right) = \sgn\det\left(E_{\ell_1}(v),E_{\ell_2}(v),E_{\ell_3}(v),\hU(v)\right)\,.
\label{determinantequal}
\ee
Since the tetrad is non--degenerate, we can find a matrix $\Omega_{v'v} \in\mathrm{GL(4)}$ for which
\be
\Omega_{v'v} E_{\ell_i}(v) = E_{\ell_i}(v')\,,\; i = 1,2,3,\qquad \Omega_{v'v} \hU(v) = \hU(v')\,.
\label{conditiononomega}
\ee
By condition \eq{cotetradmetricity}, this matrix must be orthogonal, i.e.\ $\Omega_{v'v}\in \mathrm{O(4)}$. 
From eq.\ \eq{determinantequal}, we also know that $\det \Omega = 1$, so $\Omega\in\mathrm{SO(4)}$.
Suppose now there were two matrices $\Omega_1, \Omega_2\in \mathrm{SO(4)}$ for which
\be
\Omega_1 E_{\ell_i}(v) = E_{\ell_i}(v')\,,\qquad \Omega_2 E_{\ell_i}(v) = E_{\ell_i}(v')\,,\quad i = 1,2,3\,.
\ee
This would imply that 
\be
\Omega_2^{-1} \Omega_1 E_{\ell_i}(v) = E_{\ell_i}(v')\,,\quad i = 1,2,3\,,
\ee
and hence $\Omega_2 = \Omega_1$. Therefore, the group element $\Omega_{v'v}$  is unique. \qed
\end{proof}
Together, the closure condition \eq{cotetradclosure} and equation \eq{conditionLeviCivitaconnection} can be regarded as a discrete analogue of the torsion equation $DE = 0$. 

In analogy to the continuum, we can define the concept of a tetrad. This definition makes heavy use of the duality between $\Delta$ and $\Delta^*$.
To describe the relation between tetrad and co--tetrad, it is, in fact, convenient to formulate everything in terms of the dual complex $\Delta^*$.
Each 4--simplex $\sigma$ is dual to a vertex $v$ of $\Delta^*$, i.e.\ $\sigma = v^*$. By deleting a vertex $p$ in $\sigma$, we obtain a tetrahedron $\tau$. This tetrahedron $\tau$ is, in turn, dual to an edge $e$. Thus, the choice of a 4--simplex $\sigma$ and a vertex $p\subset\sigma$ defines an edge $e$ at the dual vertex $v$. Conversely, a pair $(v,e)$ can be used to label a vertex $p$ of the triangulation.
Two different pairs $(v,e_{1})$ and $(v',e'_{1})$ correspond to the same vertex provided that 1) $(v,v') = e $ is an edge of $\Delta^{*}$ and that 
2) $(e_1, e,e'_1)$ are consecutive edges in the boundary of a face of $\Delta^{*}$.

Since vertices of the triangulation correspond to pairs $(v, e_{1})$, edges $\ell = [p_1p_2]\subset\sigma$ of $\Delta$ correspond to  triples $(v, e_{1},e_{2})$.
We can use this to translate the notation for the co--tetrad to the dual complex:  instead of denoting the co--tetrad by $E_\ell(v)$, we can write it as $E_{e_1e_2}(v)$.

 In this notation, the defining relations for the co--tetrad appear as follows:
\be
E_{ee'}(v) = -E_{e'e}(v)\,,\qquad E_{e_{1}e_{2}}(v) + E_{e_{2}e_{3}}(v) + E_{e_{3}e_{1}}(v) = 0\,.
\ee
Similarly, the equation for the spin connection becomes
\be
\Omega_e E_{e_1 e_2}(v) = E_{e'_1 e'_2}(v')\,,
\ee
where $e = (vv')$ is an edge of $\Delta^{*}$ and $(e_i, e'_i)_{i=1,2}$ are pairs of edges such that $(e_i, e, e'_i)$ are consecutive edges in the boundary of a face. Note that there are always four such pairs for a given edge $e$.

When stating relations between co--tetrad and tetrad, it is also convenient to define an orientation for each 4--simplex. 
By definition, a local orientation of $\Delta $ is a choice of $\mathbb{Z}_{2}$--ordering of vertices for each 4-simplex $\sigma$. Such an ordering is represented by tuples $[p_{1},\cdots, p_{5}]$. Two $\mathbb{Z}_{2}$--orderings that differ by an even permutation are by definition equivalent.
Two $\mathbb{Z}_{2}$--orderings that differ by an odd permutation are said to be opposite and we write $[p_{1},p_{2}\cdots, p_{5}] =  -[p_{2},p_{1}\cdots, p_{5}]$.
By duality it is clear that a local orientation is equivalent to a choice of 
 $\mathbb{Z}_{2}$--ordering $[e_{1},\cdots,e_{5}]$ of edges of $\Delta^{*}$ meeting at $v$. 
With this orientation we can also define a correspondence between edges $e_{1}$ and oriented tetrahedra $[p_{2},\cdots, p_{5}]$.

Given a choice of local orientation  of $\Delta$, one says that two neighboring 4--simplices $\sigma$, $\sigma'$ that share a tetrahedron $\tau$ are consistently  oriented, if the orientation of $\tau$ induced from $\sigma$ is \textit{opposite} to the one induced from $\sigma'$.
Namely, if $\sigma = [p_{0}, p_{1},\cdots p_{4}]$ and 
$\sigma'= -[p'_{0}, p_{1},\cdots p_{4}]$, they induce opposite orientations on the common tetrahedron $\tau=[ p_{1},\cdots p_{4}]$ and are therefore consistently oriented. The triangulation $\Delta $ is said to be \textit{orientable} when it is possible to choose the local orientations such that they are consistent for every pair of neighboring 4--simplices. Such a choice of consistent local orientations is called a global orientation.

From now on and in the rest of the paper we assume that we work with an orientable triangulation and with a fixed global orientation.
 
\renewcommand{\arraystretch}{1.8}
\begin{definition}
For a given non--degenerate co--tetrad $E$ on $\Delta$, the associated tetrad $U$ is an assignment of vectors $U_e(v)\in\bR^4$ to each vertex $v$ and
(unoriented) edge $e\supset v$ such that 
\be
U_{eI}(v) E_{e''e'}^I(v) = \delta_{e''e}- \delta_{e'e}
\label{orthogonalityrelationwithoutindices}
\ee
for all $e',e''\supset v$.
\end{definition}
These conditions specify the tetrad $U$ uniquely, as we show in appendix \ref{relationbetweencotetradandtetradina4simplex}. The orthogonality relation \eq{orthogonalityrelationwithoutindices} is the discrete counterpart of the equation $E^\mu{}_I E_\nu{}^I = \delta^\mu{}_\nu$ in the continuum.
\renewcommand{\arraystretch}{1}

Based on \eq{orthogonalityrelationwithoutindices}, we can derive a number of useful identities satisfied by a tetrad:
\begin{proposition}
\label{propositionEUmainpart}
At any vertex $v\subset\Delta^*$, the tetrad vectors $U_e(v)$ close, i.e.\ 
\be
\sum_{e\supset v} U_e(v) = 0\,.
\label{tetradclosurewithoutindices}
\ee
For a tuple $[e_1\ldots e_5]$ of edges at $v$, 
we can express the  discrete tetrad explicitly  in terms of the discrete  co--tetrad and vice versa:
\be
U_{e_2}(v) = \frac{1}{V_4(v)} \star\left(E_{e_3e_1}(v)\wedge E_{e_4e_1}(v)\wedge E_{e_5e_1}(v)\right)
\ee
and
\be
E_{e_2e_1}(v) = V_4(v) \star\left(U_{e_3}(v)\wedge U_{e_4}(v)\wedge U_{e_5}(v)\right)\,,
\label{inversionformula}
\ee
where $V_4(v)/4!$ is the oriented volume of the 4--simplex spanned by the co--tetrad vectors:
\be
V_4(v) = \det\left(E_{e_2e_1}(v),\ldots,E_{e_5e_1}(v)\right)\,.
\label{4volumeintermsofE}
\ee 
For bivectors, one has the relation
\be
\star\left(E_{e_1e_2}(v)\wedge E_{e_2e_3}(v)\right) = V_4(v)\left( U_{e_4}(v) \wedge U_{e_5}(v)\right)\,.
\ee
The norm of $U_e$ is proportional to the volume $V_3(e)/3!$ of the tetrahedron orthogonal to $U_e$:
\be
|U_{e}(v)| = \frac{V_3(e)}{|V_{4}(v)|}.
\ee
The determinant of the tetrad vectors equals the inverse of $V_4(v)$:
\be
\frac{1}{V_4(v)} = \det\left(U_{e_2}(v),\ldots,U_{e_5}(v)\right).
\ee
\end{proposition}
In this proposition  we have set
$
[\star(E_{1}\wedge \cdots \wedge E_{n})]_{I_{1} \cdots I_{4-n}} 
\equiv \epsilon_{I_{1}\cdots I_{4}} E_{1}^{I_{5-n}}\cdots E_{n}^{I_4}\,.
$ 
These statements are proven in appendix \ref{relationbetweencotetradandtetradina4simplex}.

\section{Solutions}
\label{solutions}

With the help of the previous definitions, we will now determine the solutions to the equations \eq{eqvect}, the simplicity constraints (\ref{simp1}) and the non--degeneracy condition \eq{nondeg2}. For the proofs it is practical to denote the oriented wedge $(ef)$ by an ordered pair of edges $(ee')$ which meet at $v$.
The order $(ee')$ refers to the fact that $e$ and $e'$ are consecutive w.r.t.\ the orientation of the face.
Note that the interior closure constraints $X_{ef}(v) =  X_{e'f}(v)$ mean that there is only one bivector per face $f$ and vertex $v$. Hence we can denote the bivectors by  $ X_{ee'}(v)\equiv X_{ef}(v)=X_{e'f}(v)$.

\renewcommand{\arraystretch}{1.6}
\begin{proposition}
\label{mainproposition}
Let $(j_f,n_{ef},u_e,g_{ev},h_{ef})$ be a configuration that solves eqns. \eq{eqvect}, with the bivectors defined by the simplicity condition \eq{Xintermsofjunagain}. Then, there exists a co--tetrad $E$ such that for any vertex $v$ and tuple $[e_1\ldots e_5]$ of edges at $v$
\be
X_{e_4e_5}(v) = \epsilon\,\star\left(E_{e_1e_2}(v)\wedge E_{e_2e_3}(v)\right)\,.
\label{propositionXintermsofE}
\ee
The factor $\epsilon$ is a global sign, and the 4--volume $V_4(v)$ is given by equation \eq{4volumeintermsofE}.
This co--tetrad is unique up to inversions $E_\ell(v) \to -E_\ell(v)$, $\ell\subset v^*$.  

Equivalently, the bivectors can be expressed by the associated tetrad $U$, namely, 
\be
X_{ee'}(v) = \epsilon\, V_4(v)\left(U_e(v) \wedge U_{e'}(v)\right)
\ee
for any pair of edges $e, e'\supset v$. 

Given this co--tetrad $E$ and tetrad $U$, the variables $(j_f,n_{ef},u_e,g_{ev})$ are determined as follows: the spin $j_f$ is equal to the norm of the bivector $\star  X_{ef}(v)$, and hence Regge--like. 
The group elements $g_{v'v} = g_{v'\!e} g_{ev}$ are, up to signs $\epsilon_e$, equal to the spin connection for the co--tetrad $E$, i.e.\
\be
g_{v'\!v} = \epsilon_e\,\Omega_{v'\!v},\,\qquad\epsilon_e = \pm 1\,.
\label{propositiongintermsofOmega}
\ee
For a given choice of the holonomy $g_{ev}$ on the half--edge $ev$, the group element $u_e$ is determined, up to sign, by
\be
u_e = \frac{\pm 1}{|U_e(v)|} \left(\left(g_{ev}U_e(v)\right)^0 \mathbbm{1} + \irm\left(g_{ev}U_e(v)\right)^i\sigma_i\right)\,.
\label{propositionuintermsofU}
\ee
The group element $n_{ef}$ is fixed, up to a U(1) subgroup, by
\be
n_{ef}\sigma_{3}n_{ef}^{-1}= N_{ef}^{i}\sigma_{i}\,,
\label{propositionn}
\ee
where for $f = (ee')$
\be
j_f N^i_{ef} = V_4(v) \left(g_{ev}\triangleright U_e(v) \wedge U_{e'}(v)\right)^+{}^i\,.
\ee
Conversely, every non--degenerate co--tetrad $E$ and spin connection $\Omega$ give rise to a solution via the formulas \eq{propositionXintermsofE}, \eq{propositiongintermsofOmega}, \eq{propositionuintermsofU} and \eq{propositionn}.
\end{proposition}
\begin{proof}
Let us first consider two consecutive edges $e$ and $e'$ such that $f = (ee')$. 
The simplicity condition $(\star  X_{ef}(v))\cdot \hat{U}_e(v)=0$ implies that there exists a 4--vector $N_{ef}(v)$ such that
$X_{ef}(v) = \hat{U}_e(v)\wedge N_{ef}(v)$.  Similarly there exists another vector $N_{e'f}$ such that
$X_{e'f}(v) = \hat{U}_{e'}(v)\wedge N_{e'f}(v)$.
The interior closure constraint $X_{ef}(v)=X_{e'f}(v)\equiv X_f(v)$ requires that $U_{e'}(v)$ belongs to the plane spanned by $U_e(v)$ and $N_{ef}(v)$, so there exist coefficients $a_{ef}, b_{ef}$ such that 
\be
\hat{U}_{e'}(v) = a_{ef} N_{ef}(v) + b_{ef}\hat{U}_e(v)\,.
\ee
If $a_{ef}=0$, this means that $\hat{U}_e(v) = \hat{U}_{e'}(v)$, since $\hat{U}$ are normalized vectors. This is excluded by our condition of non--degeneracy, since one would have $X_{ef_2}(v)\wedge X_{e'\!f'_2}(v)=0$ if $f_2=(ee_2)$ and $f_2'=(e'e'_2)$. Denoting $\alpha_{ee'}\equiv a_{ef}^{-1}$ one therefore has 
$N_{ef} = \alpha_{ee'} \hat{U}_{e'} - \alpha_{ee'} b_{ef}\hat{U}_e$ and hence
\be
X_{ee'}(v) = \alpha_{ee'}(v) \left(\hat{U}_e(v)\wedge\hat{U}_{e'}(v)\right)\,.
\label{beforefactorization}
\ee
It follows from this expression and the non--degeneracy condition \eq{nondeg2} that the vectors $\hat{U}_e(v)$, $e\supset v$, span a 4--dimensional space.
As shown in appendix \ref{reconstructionof4simplex}, the exterior closure constraints 
\be
\sum_{f\supset e} \epsilon_{ef}(v) X_{f}(v) = 0\,,
\ee
imply the factorization $\alpha_{ee'} = \epsilon(v) \alpha_e(v)\alpha_{e'}(v)$, where $\epsilon(v) = \pm 1$ and $\alpha_e(v)$ are real numbers, independent of the orientation of $e$, such that 
\be
\sum_{e \supset v} \alpha_e(v) \hat{U}_{e}(v) = 0\qquad\mbox{and}\qquad
j_f^2 = \alpha_e^2(v) \alpha_{e'}^2(v) \sin^2 \theta_{ee'}(v)\,.
\label{condalpha}
\ee
The angle $\theta_{ee'}(v)$ is defined by $\cos\theta_{ee'}(v) = \hat{U}_{e}(v) \cdot \hat{U}_{e'}(v)$. These conditions only admit a solution if there exists a discrete, geometrical 4--simplex (i.e.\ a set of edge lengths $\ell(v)$) such that $j_f$ and $\theta_{ee'}$ are the areas and dihedral angles in this 4--simplex.
In this case, $|\alpha_{e}(v)|$ is uniquely determined by the spins $j_f$ and the unit vectors $\hat{U}_e(v)$. The $\alpha_e(v)$ themselves are only fixed up to an overall sign, i.e.\ if $\alpha_e$, $e\supset v$, solves \eq{condalpha}, then $\epsilon_{v} \alpha_e(v)$, $\epsilon_{v}=\pm1$, is a solution as well. 

Given the ordering $e_1,\ldots, e_5$ of edges at $v$, we then define
\be
V_4(v) \equiv \det\left(\alpha_{e_2}\hat{U}_{e_2}(v),\ldots,\alpha_{e_5} \hat{U}_{e_5}(v)\right)\qquad\mbox{and}\qquad
U_e(v) \equiv \frac{\alpha_e(v)}{\sqrt{|V_{4}|}} \hat{U}_e(v)\,.
\ee
These vectors have the property that
\be
\sum_{e\supset v}  U_{e}(v) = 0\qquad\mbox{and}\qquad 
X_{ee'}(v) = \epsilon(v) V_{4}(v) \left(U_{e}(v)\wedge U_{e'}(v)\right)\,,
\label{proofXintermsofU}
\ee
where $\epsilon(v) = \pm 1$. Thus, the $U_{e}(v)$ define tetrad vectors for the 4--simplex dual to $v$, and we can use formula \eq{inversionformula} to specify corresponding co--tetrad vectors $E_\ell(v)$.

Next we need to analyze the equations that relate neighboring 4--simplices $v$ and $v'$, connected by the edge $e = (vv')$:
\be
g_{vv'} \hat{U}_{e}(v') = \hat{U}_{e}(v)\,,\qquad g_{vv'}\triangleright X_f(v') = X_f(v)\,.
\label{proofrelationbetweenvertices}
\ee
The first condition leads to $g_{vv'} U_e(v')/|U_e(v')| = \tilde{\epsilon}_e U_e(v)/|U_e(v)|$, where $\tilde{\epsilon}_e \equiv \sgn\,\alpha_e(v)\,\sgn\,\alpha_e(v') = \pm 1$. By combining this with the second condition we find that for every edge $\ell$ of the tetrahedron dual to $e=(vv')$ (see appendix \ref{reconstructionofcotetradandtetrad})
\be
g_{vv'} E_{\ell}(v') = \epsilon_e E_{\ell}(v)\,,
\label{gEE}
\ee
with the sign $\epsilon_e \equiv \tilde{\epsilon}_e\,\sgn(V_4(v) V_4(v')) = \pm 1$. We see therefore that the vectors $E_\ell(v)$ satisfy the compatibility condition
(iii) in the definition of a co--tetrad, and hence they specify a co--tetrad on the entire simplicial complex.
Equation \eq{gEE} shows furthermore that $g_{vv'}$ is, up to the sign $\epsilon_e$, equal to the spin connection $\Omega_{vv'}$ associated with the co--tetrad $E$:
\be
g_{vv'} = \epsilon_e\Omega_{vv'}\,.
\label{proofgOmega}
\ee
In appendix \ref{reconstructionofcotetradandtetrad}, we also derive that the signs $\epsilon(v)$ in eq.\ \eq{proofXintermsofU} are constant, i.e.\ $\epsilon(v) = \epsilon(v')$ for neighbouring vertices $v$ and $v'$. 

The aforementioned ambiguity in the factors $\alpha_e(v)$ is transported into the co--tetrad and tetrad: for a given solution $(j_f,n_{ef},u_e,g_{ev},h_{ef})$, the tetrad and co--tetrad are fixed up to a reversal of edges in the geometrical 4--simplices: i.e.\ up to recplacing $(E_\ell(v),U_e(v)) \to (-E_\ell(v),-U_e(v))$ for all $\ell\subset v^*$, $e\supset v$.

If we start, conversely, from a co--tetrad $E$ and its tetrad $U$, it is clear that equations \eq{proofXintermsofU} and \eq{proofgOmega} define bivectors and connections which solve the equations \eq{eqvect}. The associated spins $j_f$ and group variables $u_e$ and $n_{ef}$ follow directly from the definitions \eq{Xintermsofjunagain} and \eq{relationuUhat}. \qed 
\end{proof}

\subsection{Determination of $h$}

So far we have determined $(j_f,u_e,n_{ef},g_{ev})$ in terms of a co--tetrad $E$ and signs $\epsilon$ and $\epsilon_e \equiv \e^{\irm\pi n_e}$. In order to complete the characterization of the solution, we also need to determine $h_{ef}$. This is done in the following
\begin{proposition}
For a non--degenerate co--tetrad $E$ and a choice of global sign $\epsilon$ and edge signs $\epsilon_e$, the holonomy of $g_e = \epsilon_e\Omega_e$ around a face $f$ with starting point $v$ has the form
\be
G_f(v) = \e^{\epsilon\Theta_{f} \hat{X}_{f}(v)} \e^{\pi(\sum_{e\subset f} n_{e})\star  \hat{X}_{f}(v)}\,,
\ee
where the bivector $X_f(v)$ is determined by $E$ as in equation \eq{propositionXintermsofE} and $\hat{X}_{f}= X_f/|X_f|$.
In this equation the bivector is treated as an antisymmetric map acting on $\mathbb{R}^{4}$ and we take an exponential of this map.

In the corresponding solution $(j_f,u_e,n_{ef},g_{ev},h_{ef})$ of eqns.\ \eq{eqvect}, the group elements $h_{ef}$ are uniquely determined, up to gauge transformations, by a choice of angles $(\theta_{ef}, \tilde{\theta}_{ef})$ for each wedge: these angles are subject to the conditions
\be
\sum_{e\subset f} \theta_{ef} = \epsilon\Theta_f\,,\qquad
\sum_{e\subset f}\tilde{\theta}_{ef} = \pi \sum_{e\subset f} n_e\,.
\ee
where $\Theta_{f}$ is the deficit angle of the spin connection.
The associated wedge holonomies equal
\be
G_{ef}(v) = \e^{\theta_{ef}\hat{X}_{f}(v)} \e^{\tilde{\theta}_{ef} \star  \hat{X}_{f}(v)}
\ee
\end{proposition}
\begin{proof}
In order to do the analysis it is convenient to change the frame and base all our quantities at the center of the face $f$ (see \fig{choiceofinteriorconnection}). That is, we define
\be
X_{ef}(f) \equiv h_{fe}\triangleright X_{ef}\,,\qquad
G_{ef}(f) \equiv h_{fe}G_{ef} h_{ef} = h_{fe}g_{ev}g_{ve'}h_{e'f}\,.
\ee
It follows from the equations \eq{eqvect} that $X_{ef}(f) \equiv X_{f}(f)$ is independent of the wedge and that $G_{ef}(f) \triangleright X_{f}(f) = X_{f}(f)$.
The latter implies that
\be
G_{ef}(f) = \e^{\theta_{ef}\hat{X}_{f}(f)} \e^{\tilde{\theta}_{ef} \star  \hat{X}_{f}(f)}\,,
\ee
where $\hat{X}_f(f)= X_f(f)/|X_f(f)|$. The angles $\theta_{ef}$ and $\tilde{\theta}_{ef}$ have to satisfy constraints, as we will show now.

First we remark that the holonomy around the face $f$ can be written as a product of wedge holonomies
\be
\label{holwedge}
G_f(f) \equiv G_{e_1f}(f)\cdots  G_{e_nf}(f) = \left(h_{fe_{1}}g_{e_{1}v}\right) G_f(v) \left(h_{fe_1}g_{e_1v}\right)^{-1}\,.
\ee
On the right--hand side $G_f(v)$ is the face holonomy based at the vertex $v$.
We have seen that the connection of a solution satisfies $g_e = \epsilon_e \Omega_e$, where $\Omega$ is the spin connection and $\epsilon_e$ is an arbitrary sign. The defining property of the spin connection is that $\Omega_{vv'} E_{\ell}(v') = E_{\ell} (v)$ for all edges $\ell$ in the tetrahedron dual to $e = (vv')$. 

\psfrag{f}{$f$}
\psfrag{e1}{$e_1$}
\psfrag{e2}{$e_2$}
\psfrag{e3}{$e_3$}
\psfrag{v1}{$v_1$}
\psfrag{v2}{$v_2$}
\pic{choiceofinteriorconnection}{The wedge holonomies $G_{e_if}(f)$ have their starting and end point at the center of the face.}{3cm}{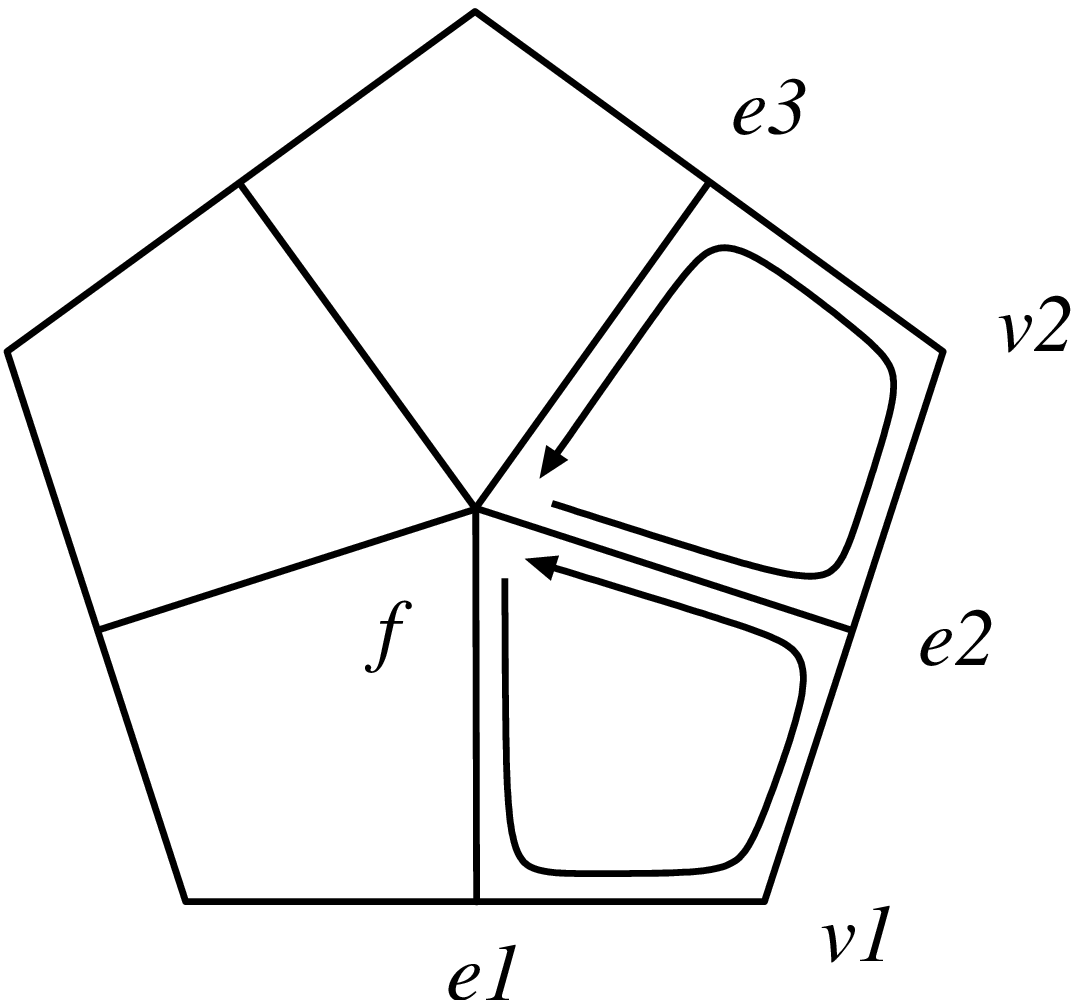}

As a result, the holonomy around the face $f$ preserves the co--tetrads associated with the triangle dual to $f$. More precisely, let us suppose that $\pa f^* = \ell_1 + \ell_2 + \ell_3$. Then, the on-shell holonomy fulfills
\be
\label{hol}
G_f(v)E_{\ell_i}(v)= \left(\prod_{e\subset f}\epsilon_{e}\right) E_{\ell_i}(v)\,,\quad i = 1,2,3\,.
\ee
As we have shown earlier, the bivector $X_f(v)$ is on-shell given by 
\be
X_f(v) = \epsilon \star \left(E_{\ell_1}(v) \wedge E_{\ell_{2}}(v)\right)\,.
\ee
Hence the condition \eq{hol} can be equivalently expressed by
\be
G_f(v) = \e^{\epsilon \Theta_{f} \hat{X}_{f}(v)} e^{\pi(\sum_{e\subset f}n_{e})  \star  \hat{X}_{f}(v)}
\ee
where $\epsilon_e \equiv \e^{\irm \pi n_{e}}$ and $\hat{X}_{f}(v)= {X}_{f}(v)/|{X}_{f}(v)|$. The angle $\Theta_f$ is the deficit angle of the spin connection w.r.t.\ the face $f$. 
Combining this result with eq.\ \eq{holwedge} one obtains that 
\be
\label{condtheta}
\sum_{e\subset f} \theta_{ef} = \epsilon \Theta_f  \,,\qquad
\sum_{e\subset f} \tilde{\theta}_{ef} = \pi \sum_{e\subset f} n_e\,. 
\ee

For a given tetrad $U$, the associated spin connection $\Omega$, and a choice of signs $\epsilon_e$, the angles $\theta_{ef}$ and $\tilde{\theta}_{ef}$ have to meet the constraint \eq{condtheta}. Once such angles $(\theta_{ef},\tilde{\theta}_{ef})$ are selected, we can solve for $h_{ef}$ recursively. For this, let us set $h_i \equiv h_{e_if}$ and define 
\be 
G_i(\theta_{e_if},\tilde{\theta}_{e_if}) \equiv \e^{\theta_{e_if}X_{f}(v_i)} \e^{\tilde{\theta}_{e_if} \star   X_{f}(v_i)}\,.
\ee
The equations
\be
G_{e_if}(v_i) = G_i(\theta_{e_if},\tilde{\theta}_{e_if})
\ee
can be recursively solved by setting
\be
h_{i+1} = g_{e_{i+1}v_i}G_i g_{v_ie_i}h_i\,,\qquad  h_1 = h_f\,,
\ee
where $h_f$ is an arbitrary initial value. This solution is consistent, since 
\bea 
\lefteqn{h_1 \equiv h_{n+1} 
=
g_{e_{1}v_{n}}G_{n}g_{v_{n}v_{n-1}}G_{n-1}\cdots g_{v_{2}v_{1}}G_{1}g_{v_1e_1} h_1} \\
&&= 
g_{e_{1}v_{n}}G_{n}(g_{v_{n}v_{n-1}}G_{n-1}g_{v_{n-1}v_{n}})\cdots(g_{v_{n}v_{n-1}}\cdots g_{v_2v_1} G_{1} g_{v_1v_2} \cdots g_{v_{n-1}v_{n}}) G^{-1}_f(v_{n}) g_{v_{n}v_1}g_{v_1e_1} h_1 \nonumber \\
&&= g_{e_{1}v_{n}} G_f(v_{n})G^{-1}_f(v_{n}) g_{v_ne_1} h_1 = h_1\,.
\eea
In the third equality, we used that $X_{ef}(v') = g_{v'\!v} X_{ef}(v) g^{-1}_{v'\!v}$.
Note that the group element $h_f$ can be fixed to the identity by a gauge transformation at the face center. This shows that, up to gauge, the elements $h_{ef}$ are determined by the choice of the angles $\theta_{ef},\tilde{\theta}_{ef}$.
\qed 
\end{proof}

\section{Semiclassical approximation of effective amplitude}
\label{semiclassicalapproximationofeffectiveamplitude}

\subsection{Evaluation of action}

In the previous section, we have seen that solutions of the equations (\ref{simp1},\ref{eqvect},\ref{nondeg2})
exist only if the set $j_{f}$ is Regge--like and, up to gauge transformation, they are uniquely determined by a 
choice of a discrete metric (coming from a co--tetrad $E$), of a global sign $\epsilon$, of edge signs $\epsilon_e$ and a choice of U(1) wedge angles $(\theta_{ef},\tilde{\theta}_{ef})$ subject to \eq{condtheta}.

\renewcommand{\arraystretch}{2}
\begin{proposition}
Given a solution characterized by the data $(U_{e},\epsilon,\epsilon_{e},\theta_{ef},\tilde{\theta}_{ef})$, 
the on-shell action is independent of $(\theta_{ef},\tilde{\theta}_{ef})$ and given by 
\be
\e^{S^{\gamma}(U_{e},\epsilon,\epsilon_{e})} = \left\{
\begin{array}{r@{\quad}l}
\e^{{{\irm \epsilon}} \sum_f  \clA_f\Theta_f} \prod_e \epsilon_e^{J_e}\,, & \gamma > 0\,, \\
\prod_e \epsilon_e^{J_e}\,, & \gamma = 0\,,
\end{array}
\right.
\label{evaluationaction}
\ee
where $\clA_f = (\gamma^+ + \gamma^-) j_f$ is the area of $f$ in Planck units $( 8\pi\hbar G=1 )$ for $\gamma > 0$ (see eq.\ \eq{areaintermsofspins}), $\Theta_f$ is the deficit angle of the spin connection, and $J_e \equiv (\gamma^{+}-\gamma^{-}) \sum_{f\supset e} j_{f}$.
\end{proposition}
Before giving the proof a few remarks are in order. Firstly, in the EPR model the on--shell evaluation is trivial, in agreement with the claim that the EPR model is a quantization of the topological sector. Moreover, for general $\gamma$, the dependence on the Immirzi parameter drops out from the on--shell action. 
\renewcommand{\arraystretch}{1}

Secondly, when evaluating the semiclassical asymptotics of the effective amplitude $W^{\gamma}_\Delta(j_f)$, one has to sum over all classical configurations and hence over $\epsilon_e$. This sum gives zero unless $J_e$ is an even integer.
It is interesting to note that when $\gamma^{\pm}$ are  both odd integers the same condition arises in the spin foam model. 

To see this, note that if  $\gamma^{\pm}$ are both odd, the condition that the weight projects down to a function of SO(4) (i.e.\ $(\gamma^{+}-\gamma^{-})j_{f}\in \mathbb{Z}$) is satisfied without any restriction on $j_{f}$, since $(\gamma^{+}-\gamma^{-})$ is even.
Moreover, the amplitudes in the spin foam model require that the invariant SU(2) subspace 
$\mathrm{Inv} \left(\otimes_{f\supset e} V_{j^{\pm}_{f}} \right)$ is non--trivial.
This is the case if and only if $\sum_{f} j^{\pm}_{f}$ is integer--valued. Therefore, $\sum_{f} j_{f}$ is integer--valued and $J_{e}$ is even.
\begin{proof} 
As shown in the previous section, the wedge holonomy has the form
\be
G_{ef}(v) = \e^{\theta_{ef} \hat{X}_{f}(v)} e^{\tilde{\theta}_{ef} \star  \hat{X}_{f}(v)}.
\ee
where $\hat{X}_f(v)= X_f(v)/|X_f(v)|$. In the SU(2)$\times$SU(2) notation this condition reads 
\be
G^{\pm}_{ef}(v) = \e^{\frac{\irm}{2} (\theta_{ef}\pm \tilde{\theta}_{ef})\hat{X}^{\pm}_f(v)}\,.
\ee
Recall also that the bivectors $X^{\gamma\pm}_f$ and $X_f^{\pm}$ are related by
\be
X^{\gamma\pm}_f = \gamma^{\pm} X_f^{\pm}\,,\qquad |X_f^{\pm}|= j_f\,.
\ee
We insert this into the action, observing that $X^{\gamma\pm}_f/ |X^{\gamma\pm}_f| = \gamma^{\pm}/|\gamma^{\pm}| \hat{X}^{\pm}_f$, and obtain 
\bea
S 
&=&
\sum_{f} \left\{ 
2|\gamma^{+}| j_{f} 
\sum_{e\subset f } \ln\left(\tr\left[\frac12\left(1+\frac{\gamma^{+}{X}^{+}_{f}}{|\gamma^{+}{X}^{+}_{f}|}\right)G_{ef}^{+}\right]\right)\right. \\
& &\left.
{} \quad\,\,\,\,\,\,\, +
2|\gamma^{-}| j_{f} 
\sum_{e\subset f } \ln\left(\tr\left[\frac12\left(1+\frac{\gamma^{-}{X}^{-}_{f}}{|\gamma^{-}{X}^{-}_{f}|}\right)G_{ef}^{-}\right]\right)
\right\} \\
&=& 
\sum_{f} \left\{
\irm\gamma^{+} j_{f} 
\sum_{e\subset f } (\theta_{ef}+\tilde{\theta}_{ef})
+ 
\irm\gamma^{-} j_{f} 
\sum_{e\subset f } (\theta_{ef}-\tilde{\theta}_{ef})
\right\} \\
&=& 
\irm \epsilon \sum_{f} (\gamma^{+} + \gamma^{-}) j_{f} \Theta_{f}
+ \irm \pi(\gamma^{+} - \gamma^{-}) \sum_e    n_e\left(\sum_{f\supset e} j_{f}\right)\,. 
\eea
\qed
\end{proof}

\subsection{Asymptotic approximation}

In order to arrive at our final result we need to determine the asymptotic approximation of the effective amplitude \eq{Zeff} for large spins. 
As shown, this partition function can be expressed as an integral 
\be
I_{N} = \int \d x\; \e^{- N S(x)}
\ee
over a set of compact variables $x$. In our case the variables are group elements, so $S$ can be taken to be a periodic function.
It is customary to restrict the study of this type of integral to the case, where $S$ is pure imaginary and use the stationary phase approximation.
It is less well-know, but nevertheless true, that the stationary phase method is valid when $S$ is a complex function, provided $\Re(S)\geq 0$ (see \cite{asymptotic} Chapter 7.7). In this reference, it is shown that when $S$ is $C^{\infty}$
and if $ |S'|^{2} + \Re(S)$ is always strictly positive (with $|S'|^{2}=  \partial_{\mu}\bar{S}\partial^{\mu} S$), then the integral is exponentially small. 
More precisely, if $S$ is $C^{k+1}$ there exists a constant $C$ such that 
\be
I_{N} \leq \frac{C}{N^{k}} \frac{1}{\mathrm{min}\left(|S'|^{2} + Re(S)\right)^{k}}\,.
\ee 
This shows that the integral is exponentially suppressed as long as $S' \neq 0 $ \textit{or} $\Re(S) > 0$.

Therefore, the dominant contribution comes from configurations that are \textit{both} stationary points of the action $S$, and absolute minima of its real part \cite{asymptotic}. One says that $x_{c}$ is a generalized critical point if 
$|S'|^{2}(x_{c}) + \Re(S)(x_{c}) =0$. In case there are such points, we have the 
asymptotic approximation
\be
I_N
\sim 
\sum_{x_c}\; \left(\frac{2\pi}{N}\right)^{\frac{r}{2}} \frac{\e^{- N S(x_c)}}{\left({\rm det}_{r}(H')\right)^{\frac12}}\,\,.
\label{wellknown}
\ee
where $x_c$ are the stationary points of $S$, $r$ is the rank of the Hessian $H=\partial_{i}\partial_{j}S(x_{c})$, $H'$ is its invertible restriction on  
$\ker(H)^{\bot}$ and $\sigma$ is the signature of $H'$. When the stationary points are not isolated, one has an integration over a submanifold of stationary points whose dimension equals the dimension of the kernel $\ker(H)$. Note that for a generalized critical point the action $S(x_{c})$ is purely imaginary.

In our case, we have shown that the effective amplitude has no generalized critical points if 
$j_{f}$ is not Regge--like. Then, the previous theorem implies that the amplitude is exponentially suppressed. When $j_{f}$ is  Regge--like, there is, up to gauge--transformations, a discrete set of solutions labelled by $(E(j_{f}),\epsilon_{e}, \epsilon)$. This result is only valid if one restricts the integration to non--degenerate configurations $|X\wedge X| >\alpha$, with $\alpha$ an arbitrary small positive number.

When applied to the integral \eq{Zeff}, this gives us that
\be
W^{ND_{\alpha}\gamma}_{\Delta }(Nj_f) \quad\sim\quad
\frac{c_{\Delta}(j_f)}{\sqrt{N}^{r_\Delta}} \sum_{\epsilon,\epsilon_{e}} 
\exp\left(N S_{\Delta}^{\gamma}(E(j_{f}),\epsilon,\epsilon_{e})\right)\,
= \frac{c_{\Delta}(j_f)}{\sqrt{N}^{r_\Delta}}\left(\exp\left(N S_R\right) + \mbox{c.c.}\right)\,,
\ee
if the set $j_f$ is Regge--like and all $J_e$ even. Otherwise the amplitude is exponentially suppressed. 
If there are several tetrad fields $E(j_{f})$ that correspond to a given set $(j_{f})_{f}$ one should also sum over them.

While we have not computed the Hessian, our analysis can give us explicit information about its rank $r_\Delta$.
In our case, the space of integration is the space of $(u_{e},n_{ef},\gb_{ev},\hb_{ef})$ which is of dimension $D = 3 E + 2 W + 6\times 2 E + 6 W$.
Here, $E,W, F$ and $V$ denote the number of  edges, wedges, faces and vertices of $\Delta^*$.
As we have seen, the space of solutions is labelled by gauge transformations $(\lambda_{e},\lambda_{f},\lambda_{v})$ and two U(1) angles $(\theta_{ef},\tilde{\theta}_{ef})$ subject to one constraint per face. Thus, the dimension of the kernel of $H$ is 
$d = 6E + 6F + 6V + 2W -2F$. We can then compute the rank to be 
\be
r_{\Delta} \equiv D - d = 33E - 6V - 4F\,,
\ee 
using the fact that $W = 4E = 10V$.

\subsection{Degenerate sector}

In order to complete our analysis of the effective amplitude and show its asymptotic Regge-like behavior, we have restricted the summation to  
non--degenerate configurations.

One could wonder wether the degenerate contributions are dominant or subdominant in this semiclassical limit\footnote{For instance, in the analysis of the 
$10j$--symbol it was shown that the degenerate configurations were non--oscillatory, but dominant \cite{FLouapre6j, Barrettas}}. 
This amounts to asking which sector has the most degenerate Hessian, since the amplitude is suppressed by $1/N^{\frac12}$ to the power of the rank of the Hessian.
Thus, it is the sector with the higher--dimensional space of solutions (higher dimensional phase space) that dominates, or in other words the one with higher entropy.
 
In order to get an idea of the dimension of the space of solutions in both sectors, let us look at the solution of the simplicity and closure constraints at a single vertex. In the non--degenerate sector, it is given by 
\be
X_{ij} = V U_i\wedge U_j\,,\quad \sum_i U_i = 0\,.
\ee
This describes $10$ rotationally invariant degree of freedom, counting $5\times 4$ $U$'s subject to $4$ independent constraints minus $6$ rotations.
These 10 degrees of freedom match the $10$ area spins.

On the other hand of the spectrum we can look at the most degenerate contribution, where all the wedge products of $X$'s are zero.
In this case, the most degenerate solution is given by 
\be
X_{ij} = U\wedge N_{ij}\,,\quad \sum_i N_{ij} = 0\,,\quad U^{2}=1\,,\quad N_{ij}\cdot U = 0\,.
\ee
Due to the last equation, the $N_{ij}$ are, in effect, 3--dimensional vectors.
Now, the counting of rotationally invariant degrees of freedom gives $15$, 5 more degrees of freedom per vertex than in the non degenerate case.
Indeed, we have $3$ U's plus $3\times 10 $ $N$'s minus $4\times 3$ independent constraints minus $6$ rotations.
 
For each $i$ we can reconstruct a geometrical tetrahedron from $N_{ij}, j\neq i$, for which $N_{ij}$ are the area normal vectors. Hence the degenerate solution determines 5 tetrahedra. These 5 tetrahedra are ``glued together'' in the sense that the faces shared by tetrahedra have the same area.
However, they do not form a 4--simplex. In a 4--simplex the volume of each tetrahedron is fixed by the area of the faces, while in the degenerate case the 5 3--volumes are independent variables and thus increase the phase space dimension.
 
This argument indicates that the phase space dimension of the degenerate configurations is higher than the non--degenerate one by at most $5$ times the number of vertices. This result bears some similarity with the recent canonical analysis of \cite{JBianca}, where it was pointed out that the phase space dimension associated with spin networks is higher than the corresponding dimension for discrete geometries. Our reasoning suggests that this extra phase space corresponds to 4d degenerate solutions.

This analysis is suggestive, but not complete, since one would need to analyze the gluing equations and the other degenerate sectors.
However, it leads one to suspect that the degenerate configuration dominate the effective amplitude in the semiclassical limit if they are included. In this case, the non--degeneracy requirement would be necessary. One challenge is to be able to formulate this requirement at the level of the spin foam model and not only in the path integral representation.
Another possibility is that the degenerate contributions are suppressed when we couple the effective amplitude to a semiclassical boundary state. This is a scenario that has been realized in the case of the $10j$--symbol \cite{gravitonpropagatorLQG}.

\section{Summary and discussion}

In this work, we have studied the semiclassical properties of the Riemannian spin foam models FK$\gamma$.
We have shown that, in the semi-classical limit, where all the bulk spins are rescaled, the amplitude converges rapidly towards the exponential of $\irm$ times the Regge action, provided the face's spins can be understood as coming from a discrete geometry.
When the spins do not arise from a discrete geometry, the spin foam amplitude is exponentially suppressed.

There are several remarks to be made about this result:
First, it is shown for an \textit{arbitrary} triangulation and not only for the amplitude associated with a 4--simplex.
This should be contrasted with what was achieved in the context of the Barret--Crane model, where only one or two 4--simplices were considered \cite{gravitonpropagatorLQG}. An extension of these results to more 4--simplices seemed increasingly complicated (see \cite{satz} for a very recent discussion of this in the context of Regge calculus).
The second fact to be noticed is that the Immirzi dependence drops out in the semiclassical limit. This should indeed be the case, since nothing depends on the Immirzi parameter at the classical level (except when it is zero). Nevertheless it was not obvious from the original definition of the amplitude that this would happen. Also, the results shown here depend heavily on the details of the implementation of the simplicity constraints: they rely on the specific choice of the measure $D_{j,k}^{\gamma}$ (see eq.\ \eq{Z}). For instance, we cannot extend our results to the ELPR$\gamma$ model for $\gamma>1$ which includes the Barret--Crane model for $\gamma=\infty$. A fourth point concerns the fact that in spin foam models areas are the natural variables, whereas 
one needs access to edge lengths in order to have a discrete geometry. To formulate constraints on areas, so that they correspond to discrete geometries, has been so far one of the conundrums faced in the LQG/spin foam approach. Several studies have been launched in order to tackle this problem (see for instance \cite{barrett, Makela, dittrichspeziale}), but the results show that performing this explicitly is an incredibly difficult algebraic task. What we find quite remarkable is that it is not necessary to answer this question analytically to get the proper semiclassical limit of a spin foam model. The spin foam model ``knows'' which set of areas does or does not arise from a 4d geometry and it naturally suppresses the non--geometric phase in the semiclassical limit.

These results provide considerable evidence in favor of the proposed spin foam amplitude as a valid amplitude for quantum gravity, in the sense that it reproduces expected semiclassical behavior. There is, however, more work to be done to fully confirm this picture.

First of all, in order to obtain this result we have to restrict the summation to non--degenerate configurations.
We know how to implement this restriction in the path integral formulation, but not in terms of the spin foam model.
As we have argued, this restriction may be important in order to get the correct semi-classical limit, but a deeper analysis is clearly required to establish
 this firmly.

More crucially, we have shown the semiclassical property of the bulk amplitude, where the bulk spins are fixed and uniformly rescaled to large values. 
That is, we have demonstrated the proper semiclassicality for certain histories that one should sum over in computing amplitudes. 
What we are ultimately interested in is the semiclassical property of the \textit{sum} over amplitudes. Given a boundary spin network, we would like to sum over all spins in the interior compatible with the boundary spin network and show that the resulting amplitude gives an object that can be interpreted as the exponential of the Hamilton--Jacobi functional of a gravity action. Our result is a necessary condition for this to happen, but we have not shown that this is sufficient.

What would be required is that for given semiclassical boundary states peaked on large spins, the corresponding amplitude is peaked around large bulk spins as well; and that the semiclassical amplitude reduces effectively to a summation over discrete geometries with the Regge action.
In a sense, one needs that the large spin limit and the integration over the spins commute with each other. 
Whether this happens or not is not obvious: one might be worried, for instance, that the summation over spins is much less restricted than a summation over discrete geometries and that this will lead to stronger equations of motions. It might be, on the other hand, that the exponential suppression of non Regge--like configurations is strong enough to effectively reduce the summation to a sum over geometries.
This is an important question that deserves to be studied further.

An obvious open problem is whether our results can be extended to the Lorentzian case. We expect that this is possible, however, it has not been shown yet wether the present Lorentzian models admit a nice action representation, which is needed for our analysis.

Moreover, our work does not address the question of the continuum limit of spin foam models.
We have considered the semiclassical limit of discrete configurations on a fixed triangulation. 
One might want to take a continuum limit, where the number of boundary vertices of the spin network grows.
It is not clear if such a limit commutes with the semiclassical limit taken here.

Despite all these open questions, we feel that the semiclassicality shown here opens the way towards new, exciting developments in the spin foam approach to quantum gravity.

\begin{acknowledgments}
We thank Fernando Barbero, Bianca Dittrich, James Ryan, Simone Speziale, Thomas Thiemann and the participants of the 
Young Loops and Foams 08 conference (where this work was presented) for many discussions.
Research at Perimeter Institute is supported by the Government of Canada through Industry Canada and by the Province of Ontario through the Ministry of Research \& Innovation. 
\end{acknowledgments} 

\begin{appendix}

\section{Relation between co--tetrad and tetrad in a 4--simplex}
\label{relationbetweencotetradandtetradina4simplex}

Based on the duality \eq{orthogonalityrelationwithoutindices} between discrete tetrad and co--tetrad, we can prove a number of identities that are analogous to equations for the co--tetrad and tetrad in the continuum. Consider a vertex $v$ in $\Delta^*$ and label the vertices $p\subset v^*$ (and corresponding dual edges $e\supset v$) by lowercase letters $i,j,k \ldots = 1,\ldots, 5$. The tetrad and co--tetrad vectors $U_e(v)$ and $E_{ee'}(v)$ are written as $U^i$ and $E_{ij}$. We denote $\mathbb{R}^4$--indices by capital letters $I,J,K$ etc.
\begin{proposition}
\label{propositionappendixtetradcotetrad}
The equation
\be
U_{I}^{i} E_{mk}{}^I = \delta_{m}^{i}-\delta_{k}^{i}.
\label{UEdelta}
\ee
determines a bijection between non--degenerate vectors $E_{ij}\in \bR^4$, $i,j = 1,\ldots, 5$, satisfying 
\be
E_{ij} + E_{ji} = 0\,,\quad\quad E_{ij} + E_{jk} + E_{ki} = 0\quad \forall\; i,j,k = 1,\ldots, 5\,,
\label{closureE}
\ee
and non--degenerate vectors $U^{i}\in \bR^4$, $i= 1,\ldots,5$, for which $\sum_{i=1}^{5} U^{i} = 0$.

The map from $E_{ij}$ to $U^{i}$ is given by
\be
U^i = \frac{1}{3!\,V_4}\,\sum_{j_{1}, j_2 , j_{3}} \epsilon^{kij_{1}j_2 j_{3}} \star\left(E_{j_1k}\wedge E_{j_2k} \wedge E_{j_{3}k}\right)\,,
\label{UintermsofE}
\ee
where $k$ is \textit{any} vertex different from $i$. $U^i$ is independent of this choice thanks to the identity \eq{closureE}.
$V_4$ denotes the oriented volume of the 4--parallelotope spanned by the co--tetrad vectors,
\be
V_4 = \det\left(E_{21},\ldots,E_{51}\right)\,,
\ee
and we have set
\be
[\star(E_{1}\wedge \cdots \wedge E_{n})]_{I_{1} \cdots I_{4-n}} 
\equiv \epsilon_{I_{1}\cdots I_{4}} E_{1}^{I_{5-n}}\cdots E_{n}^{I_4}\,.
\ee
The norm of $U^i$ is proportional to the volume $V_3$ of the tetrahedron orthogonal to $U^i$:  
\be
|U^i| = \frac{V_3}{|V_4|}
\label{tetprop}
\ee
The inverse of $V_4$ equals the determinant of the $U$'s:
\be
\frac{1}{V_4} = \det\left(U_{21},\ldots,U_{51}\right)
\label{inversevolume}
\ee
The inverse map from $U$ to $E$ is specified by
\be
E_{jk} = \frac{1}{3!}\,V_4\,
\sum_{i_{1},i_2, i_{3}}\epsilon_{k j i_1 i_2 i_3} 
\star \left(U^{i_1}\wedge U^{i_{2}} \wedge U^{i_{3}}\right)\,.
\label{EintermsU}
\ee
More generally, the relation between $U$ and $E$ is given by
\be
U^{i_1}{}_{[I_1}\cdots U^{i_n}{}_{I_n]} 
= 
\frac{1}{(4-n)!\, V_4}\, \sum_{j_1\ldots j_{4-n}} \epsilon^{ki_1\ldots i_n j_1\ldots j_{4-n}}
\star\left(E_{j_1k}\cdots E_{j_{4-n}k}\right)_{I_1\ldots I_n}\,,
\qquad k\neq i_1,\ldots, i_n\,.
\label{generalrelationUE}
\ee
The special cases $n=1$ and $n=3$ return equations \eq{UintermsofE} and \eq{EintermsU} respectively.
For $n = 2$ one obtains
\be
V_4\, U^i{}_{[I} U^j{}_{J]} = \sum_{m, n} \epsilon^{kijmn}\, \frac{1}{2}\,\epsilon^{IJ}{}_{MN}\,E_{mk}{}^M E_{nk}{}^N \,,\qquad k\neq i,j\,.
\label{specialcasen=2}
\ee
\end{proposition}
\begin{proof}
For the first part of the proof, let us assume the vectors $E_{ik}$ with property \eq{closureE} are given and that the $U_i$'s satisfy relation \eq{UEdelta}.
The identity \eq{UintermsofE} is proven, like in the continuum, by contracting the left-- and right--hand side with $E_{jk}$. That the $U_i$'s close follows directly from \eq{UintermsofE}. Formula \eq{tetprop} can be derived by using \eq{UintermsofE} and the relation between volume and Gram's determinant. Identity \eq{inversevolume} follows from \eq{UEdelta} and the multiplication rule for determinants. By contraction and use of \eq{inversevolume}, we also verify eq.\ \eq{EintermsU}.

To demonstrate that the right--hand side is independent of $k\neq i$, it helps to regard the $E_{ik}$ as edge vectors of a 4--simplex in $\bR^4$. We can think of this 4--simplex as the image of the 4--simplex $\sigma\subset\Delta$ under an affine transformation. Let $P_1,\ldots,P_5\in\bR^4$ denote the images of the vertices $p_1,\ldots,p_5\subset\sigma$.
Then, the edge vectors are equal to
\be
E_{ik} = P_i - P_k\,.
\ee
Without loss of generality, we can assume that
\be
\sum_{i=1}^5 P_i = 0\,.
\ee
Using this, we deduce that
\be
U^i = \frac{1}{3!\, V_4}\, \sum_{k, j_1, j_2, j_3} \epsilon^{kij_1 j_2 j_3} \star\left(P_{j_1}\wedge P_{j_2}\wedge P_{j_{3}}\right)\,,
\ee
making the independence of $k\neq i$ in \eq{UintermsofE} manifest.

Conversely, suppose we have vectors $U_i$ that close and that the $E_{jk}$'s fulfill relation \eq{UEdelta}.
We then define vectors 
\be
P_j = \frac{1}{5\cdot 3!}\, V_4\,\sum_{k, i_1, i_2, i_3} \epsilon_{k j i_1 i_2 i_3} \star\left(U^{i_1}\wedge U^{i_2}\wedge U^{i_3}\right)
\ee
and verify that 
\be
E_{ik} = P_i - P_k\,.
\ee
Hence the vectors $E_{ik}$ close. 

Relation \eq{generalrelationUE} is demonstrated by contracting with $n$ $E$'s.
\qed
\end{proof}

\section{Reconstruction of 4--geometry}
\label{reconstructionof4geometry}

In this appendix, we complete the proof of proposition \eq{mainproposition}. In the first part, we will derive that the bivectors $X_f(v)$ arise from a geometric 4--simplex. A key step for this is that the factors $\alpha_{ee'}$ in eq.\ \eq{beforefactorization} factorize. 
In the second part, we derive relations among tetrad vectors between neighbouring vertices, showing that the tetrad and co--tetrad vectors define a consistent discrete geometry on the simplicial complex. We will prove, in particular, that the sign factors $\epsilon(v)$ in eq.\ \eq{proofXintermsofU} are the same for every vertex.

\subsection{Reconstruction of 4--simplex}
\label{reconstructionof4simplex}

Consider a vertex $v\subset\Delta^*$ and the edges $e_1,\ldots, e_5\supset v$. To simplify formulas, we use the abbreviations $X_{ij} \equiv X_{e_ie_j}$, $U_i \equiv U_{e_i}$ and $E_{ij} \equiv E_{e_ie_j}$.
\begin{proposition}
Let $X_{ij} = -X_{ji}$, $i,j=1,\ldots,5$, be non--degenerate bivectors (i.e.\ $|X_{ij}\wedge X_{kl}| > 0$) which satisfy the simplicity and closure constraint 
\bea
&& X_{ij}^{IJ}(\hU_{i})_{J} = 0\,, 
\label{simplicityconstraintinindexnotation} \\ 
&& \sum_{j\neq i} X_{ij} = 0\,.
\eea
Then, there are, modulo translations, precisely two 4--simplices whose area bivectors equal $\star X_{ij}$ and they are related by a reversal of edge vectors. That is, there are exactly two sets of vectors $E_{ij}\in \bR^4$, $i,j = 1,\ldots,5$, obeying the closure condition \eq{closureE}, such that
\be
X_{ij} = \epsilon \sum_{m,n} \frac{1}{2}\,\epsilon_{kijmn} \star\left(E_{mk}{}\wedge E_{nk}\right)\,,\quad k\neq i,j\,.
\ee
The sign factor $\epsilon$ is either $1$ or $-1$ $\forall\; i,j = 1,\ldots 5$. The two sets $\{E_{ij}\}$ are related by the SO(4) transformation $E_{ij}\to 
-E_{ij}$.
\label{reconstructionof4s}
\end{proposition}
\begin{proof}
The simplicity constraints \eq{simplicityconstraintinindexnotation} imply that 
\be
X_{ij} = \alpha_{ij} \hU_i \wedge \hU_j\,,
\ee
where $\alpha_{ij}$ is a symmetric matrix of normalization factors and the wedge product stands for the bivector
\be
\left(\hU_i \wedge \hU_j\right)^{IJ} = \hU_i{}^{[I} \hU_j{}^{J]} = \hU_i{}^I \hU_j{}^J - \hU_j{}^I \hU_i{}^J\,.
\ee
The closure constraint states that
\be
\sum_{j\neq i} \alpha_{ij} \hU_i \wedge \hU_j = \hU_i \wedge \sum_{j\neq i} \alpha_{ij} \hU_j = 0\quad \forall\; i = 1,\ldots, 5\,. 
\ee
Consequently,
\be
\sum_{j=1}^5 \alpha_{ij} \hU_j = 0
\label{alphaijUj}
\ee
for suitable diagonal elements $\alpha_{ii}$. 

Next we eliminate one of the five $\hU_j$ in the last equation, say, $\hU_m$. For arbitrary $k,l$, $k\neq l$,
\be
\sum_j \left(\alpha_{km}\alpha_{lj} - \alpha_{lm}\alpha_{kj}\right) \hU_j 
=
\sum_{j\neq m} \left(\alpha_{km}\alpha_{lj} - \alpha_{lm}\alpha_{kj}\right) \hU_j 
= 0\,.
\ee
Since the bivectors are non--degenerate, four of the five normal vectors $\hU_i$ must be linearly independent.
Therefore,
\be
\alpha_{km} \alpha_{lj} = \alpha_{kj} \alpha_{lm}\,.
\ee
In particular, for $l = j$,
\be
\alpha_{km} \alpha_{jj} = \alpha_{kj} \alpha_{jm}\,.
\ee
By non--degeneracy, all $\alpha_{ij}$ are non--zero, so
\be
\alpha_{km} = \frac{\alpha_{kj} \alpha_{jm}}{\alpha_{jj}} = \frac{\alpha_{kj} \alpha_{mj}}{\alpha_{jj}}\,.
\ee
Let us pick one $j = j_0$ and define
\be
\alpha_i \equiv \frac{\alpha_{ij_0}}{\sqrt{|\alpha_{j_0j_0}|}}\,.
\ee
Then,
\be\label{alpha}
\alpha_{ij} = \sgn(\alpha_{j_0j_0})\,\alpha_i \alpha_j
\ee
and the bivectors have the form
\be
X_{ij} = \tilde{\epsilon} \left(\alpha_i\hU_i\right)\wedge\left(\alpha_j\hU_j\right)\,,
\label{afterfactorization}
\ee
where $\tilde{\epsilon} = \sgn(\alpha_{j_0j_0})$ is a sign independent of $i$ and $j$. From eq.\ \eq{alphaijUj} we also know that
\be
\sum_j \alpha_j\hU_j = 0\,.
\label{sumajUj}
\ee
By taking the square of eq.\ \eq{afterfactorization}, we get
\be
j^2_{ij} = \alpha_i^2 \alpha_j^2 \sin^2 \theta_{ij}\,,\qquad \cos\theta_{ij} = \hat{U}_i \cdot \hat{U}_j\,,
\ee
which fixes the modulus of $\alpha_i$ given $j_{ij}$ and $\hat{U}_{i}$. Eq.\ \eq{sumajUj} implies furthermore that the signs $\sgn\,\alpha_i$ are fixed up to an overall 
sign change $\alpha_i\to -\alpha_i$, $i = 1,\ldots, 5$.

At this point, we can reconstruct the tetrad and co--tetrad vectors. Define
\be
U_i \equiv \frac{\alpha_i\hU_i}{\sqrt{|V_4|}} \quad\mbox{with}\quad  
V_4 \equiv \det\left(\alpha_{2}\hU_2,\ldots,\alpha_{5} \hU_5\right) \,.
\ee
Then, we obtain that
\be
\frac{1}{V_4} = \det\left(U_2,\ldots,U_5\right)\ee
and
\bea
X_{ij} = \tilde{\epsilon}\,|V_4|\,U_i\wedge U_j = \epsilon\,V_4\,U_i\wedge U_j\,,
\label{XijUiUj}
\eea
where $\epsilon \equiv \tilde{\epsilon}\, \sgn(V_{4})$.
By proposition \ref{propositionEUmainpart} and \ref{propositionappendixtetradcotetrad}, the $U_i$'s define corresponding dual vectors $E_{ij}$ such that
\be
X_{ij} = \epsilon \sum_{m,n} \frac{1}{2}\,\epsilon_{kijmn} \star\left(E_{mk}{}\wedge E_{nk}\right)\,,\quad k\neq i,j\,.
\label{XijEmkEnk}
\ee
\qed
\end{proof}

\subsection{Reconstruction of co--tetrad and tetrad}
\label{reconstructionofcotetradandtetrad}


Next we deal with the equations \eq{proofrelationbetweenvertices} that relate variables from neighbouring 4--simplices.
We consider an edge $e=(vv')$ and employ the following shorthand notation:
\bea
U_0\equiv U_e(v)\,, & \qquad & U_0'\equiv g_{vv'} U_e(v')\,, \\
U_i\equiv U_{e_i}(v)\,, & \qquad & U_i' \equiv g_{vv'} U_{e_i'}(v')\,, \\
E_{ij}\equiv E_{e_ie_j}(v)\,, & \qquad & E'_{ij} \equiv g_{vv'} E_{e_ie_j}(v')\,. 
\eea
The labels $i$ are chosen such that $(e_i\,e \,e'_i)$ corresponds to one of the four faces adjacent to $e$ (see fig. \ref{consistent}).
One can check that this ordering is compatible with our requirement that orientations of neighbouring 4--simplices are consistent.

\psfrag{1}{$e_1$}
\psfrag{2}{$e_2$}
\psfrag{3}{$e_3$}
\psfrag{4}{$e_4$}
\psfrag{1'}{$e'_1$}
\psfrag{2'}{$e'_2$}
\psfrag{3'}{$e'_3$}
\psfrag{4'}{$e'_4$}
\psfrag{e}{$e$}
\psfrag{v}{$v$}
\psfrag{v'}{$v'$}
\pic{consistent}{Choice of labelling at neighbouring vertices $v$ and $v'$.}{2.5cm}{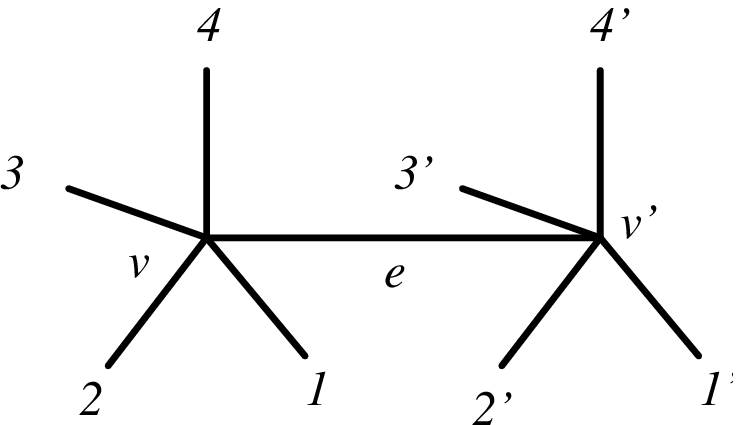}

As seen in section \ref{solutions}, the exterior closure constraints leads to
\be
\label{closi} 
\sum_i U_i = -U_0\quad\mbox{and}\quad \sum_i U_i' = -U_0'\,.
\ee
Moreover, due to the eqns.\ \eq{proofrelationbetweenvertices}, the $U$ and $U'$ are related as follows:
\be
\label{relationU}
\frac{U_0'}{|U_0'|} = \tilde{\epsilon} \frac{U_0}{|U_0|}\,, \qquad 
X_{0i} = \epsilon V (U_0\wedge U_i) = \epsilon' V' (U_0'\wedge U_i')\,,
\ee
where $\epsilon, \epsilon', \tilde{\epsilon} = \pm 1$ and 
\be 
1/V \equiv \det\left(U_{1},U_{2},U_{3},U_{4}\right)\,,\qquad 1/V'\equiv \det\left(U_{1}',U_{2}',U_{3}',U_{4}'\right).
\ee
\begin{proposition}
The conditions (\ref{closi}, \ref{relationU}) imply that 
\be
\label{relUU'} 
\epsilon = \epsilon'\,,\quad 
\tilde{\epsilon} = \sgn(VV')\alpha\,,\quad
V U_0 = \alpha\, V' U_0' \quad\mbox{and}\quad 
U_i' = \alpha U_i + a_i U_0\,,
\ee 
where $\alpha$ is an arbitrary sign factor and $a_i$ are coefficients such that $\sum_i a_i = \alpha \left(1 - \frac{V}{V'}\right)$.
Moreover, for the co--tetrad vectors $E_{ij}$ and $E'_{ij}$ one has the identity
\be
E_{ij}'=\alpha E_{ij}\,.
\ee
\end{proposition}
\begin{proof}
The equations \eq{relationU} tell us that $U_0'$ is proportional to $U_0$ and that $U_i'$ is a linear combination of $U_i$ and $U_0$. More 
precisely,
\be
\label{UU'} 
U_i' =  \tilde{\epsilon}\epsilon\epsilon' \frac{|U_0| V}{|U_0'| V'} U_i + a_i U_0\,,
\ee 
where $a_i$ are coefficients such that $\sum_i U_i' = -U_0'$. It follows that
$\sum_i a_i - \tilde{\epsilon} \epsilon\epsilon' \frac{|U_0| V}{|U_0'| V'} = -\tilde{\epsilon} \frac{|U_0'|}{|U_0|}$.
Using the relation \eq{UU'}, we obtain
\bea
\quad 1/V' &\equiv& {\rm det}\left(U_{1}',U_{2}',U_{3}',U_{4}'\right)
= {\rm det}\left(U_{0}',U_{1}',U_{2}',U_{3}'\right)\\
&=& \tilde{\epsilon} \frac{|U_{0}'| }{|U_{0}| }\left( \tilde{\epsilon}\epsilon\epsilon' \frac{|U_{0}| V}{|U_{0}'| V'}\right)^{3}
{\rm det}\left(U_{0},U_{1},U_{2},U_{3}\right)
= \epsilon\epsilon' \left( \frac{|U_{0}| V}{|U_{0}'| V'}\right)^{2} {1}/{V'}.
\eea
Thus, $\epsilon=\epsilon'$ and $|U_0| V = \pm |U_0'| V'$. By defining the sign factor
\be
\alpha \equiv \tilde{\epsilon}\frac{|U_0| V}{|U_0'| V'}\,,
\ee
we arrive at eq.\ \eq{relUU'}.

By using eq.\ \eq{EintermsU} of prop.\ \eq{propositionappendixtetradcotetrad}, we can now compute explicitly the relation between co--tetrad vectors for edges that are shared by the 4--simplices dual to $v$ and $v'$:
\be
E_{jk}' = \frac{1}{3!} V' \epsilon_{jk}{}^{i_1i_2i_3} \star\left(U_{i_1}'\wedge U_{i_2}'\wedge U_{i_3}'\right) 
= \alpha^3 V \epsilon_{jk}{}^{i_1i_2i_3} \star\left(U_{i_1}\wedge U_{i_2}\wedge U_{i_3}\right) 
= \alpha E_{jk}\,.
\ee
\qed
\end{proof}
The relation $E_{ij}' = \alpha E_{ij}$ shows that the $E_{ij}$ satisfy the metricity condition (iii) in the definition of a co--tetrad.
Therefore, the co--tetrad and tetrad vectors determine a consistent 4--geometry on the simplicial complex.
 
\end{appendix}

\end{document}